\documentclass{article}

\usepackage{arxiv}

\usepackage[utf8]{inputenc} 
\usepackage[T1]{fontenc}    
\usepackage{hyperref}       
\usepackage{url}            
\usepackage{booktabs}       
\usepackage{amsfonts}       
\usepackage{nicefrac}       
\usepackage{microtype}      
\usepackage{amssymb}        
\usepackage{lipsum}
\usepackage{subfigure}
\usepackage{graphicx}
\usepackage{marvosym}
\usepackage{wrapfig} 
\usepackage{amsmath}
\usepackage{empheq}
\usepackage{subcaption}
\graphicspath{ {./images/} }

\usepackage{xcolor}
\hypersetup{hidelinks,
	colorlinks=true, linkcolor=purple, urlcolor=teal, citecolor=cyan!60!black,
	pdfstartview=Fit,
	breaklinks=true}

\usepackage{amsthm}
\newtheorem{theorem}{Theorem}
\newtheorem{lemma}[theorem]{Lemma}
\theoremstyle{definition}
\newtheorem{definition}{Definition}[subsection]

\usepackage{fancybox}
\usepackage{tikz}
\newcommand*\circled[1]{\tikz[baseline=(char.base)]{
            \node[shape=circle,draw,inner sep=1pt] (char) {\scriptsize #1};}}

\usepackage{multirow}
\usepackage{multicol}
\usepackage{booktabs}
\usepackage{appendix}

\title{Lmte: Putting the “Reasoning” into WAN Traffic Engineering with Language Models\thanks{Correspondence to: Yan Qiao and Zonghui Wang. This work (accepted by IEEE INFOCOM 2026 - IEEE Conference on Computer Communications) raises no ethical concerns. All network traffic matrices used in this paper contain no personally identifiable information (PII).}}

\author{
 Xinyu Yuan \\
  College of Computer Science and Technology\\
  Zhejiang University\\
  Hangzhou, China\\
  \texttt{yxy5315@gmail.com} \\
   \And
 Yan Qiao\textsuperscript{\Letter} \\
  School of Computer Science and Infomation Engineering\\
  Hefei University of Technology\\
  Hefei, China\\
  \texttt{qiaoyan@hfut.edu.cn} \\
  \And
 Zonghui Wang\textsuperscript{\Letter} \\
  College of Computer Science and Technology\\
  Zhejiang University\\
  Hangzhou, China\\
  \texttt{zhwang@zju.edu.cn} \\
  \And
 Meng Li \\
  School of Computer Science and Infomation Engineering\\
  Hefei University of Technology\\
  Hefei, China\\
  \texttt{mengli@hfut.edu.cn} \\
  \And
 Wenzhi Chen \\
  College of Computer Science and Technology\\
  Zhejiang University\\
  Hangzhou, China\\
  \texttt{chenwz2@zju.edu.cn} \\
}

\begin{document}
\maketitle

\begin{abstract}
The rapid expansion of modern wide-area networks (WANs) has made traffic engineering (TE) increasingly challenging, as traditional solvers struggle to keep pace. Although existing offline ML-driven approaches accelerate TE optimization with deep neural networks (DNNs), they often lack sufficient expressiveness and generalization on unseen traffic patterns or topologies, limiting their practicality. Inspired by the success of large language models (LMs), for the first time, this paper investigates their potential as general-purpose traffic planners. Our contributions are two-fold: 
(i) Theoretically, we show that pre-trained LMs can simulate the sequential decision processes underlying TE and, crucially, exhibit parallel reasoning capabilities, making them well-suited for the task; (ii) Practically, we present \textsc{Lmte}, a novel LM-driven TE framework that embraces these insights through efficient multimodal alignment and lightweight configuration generation, all while preserving the model's original abilities. Extensive experiments demonstrate that \textsc{Lmte} matches top-tier performance on five datasets, achieving up to 15\% better  maximum link utilization (MLU) and consistently lower performance degradation across diverse scenarios, e.g., less than 5\% with high traffic dynamics and link failures. Moreover, it achieves 10 to 100 times speedups over traditional TE solvers. To aid future works, our codebase is available at \href{https://github.com/Y-debug-sys/LMTE}{\texttt{https://github.com/Y-debug-sys/LMTE}}.
\end{abstract}


\section{Introduction}
\label{sec:introduction}

Modern wide-area networks (WANs) interconnect geographically distributed data centers using high-capacity optical links, forming a critical yet expensive physical infrastructure~\cite{jain2013b4,poutievski2022jupiter,abuzaid2021contracting}. To ensure high availability and low latency, traffic engineering (TE) plays a vital role in managing network performance. Typically planned by a centralized Software-Defined Networking (SDN) controller, TE periodically solves mathematical optimization problems to route traffic efficiently in the face of dynamic topology constraints and fluctuating service demands. This topic has been extensively studied across a variety of network environments~\cite{abuzaid2021contracting,narayanan2021solving,wang2006cope,kumar2018semi,applegate2003making,perry2023dote,xu2023teal,wang1999explicit,alqiam2024transferable}.

\begin{wrapfigure}[16]{r}{0.5\textwidth}
  \vspace{-0.25cm}
  \begin{center}
  \includegraphics[width=1.\linewidth]{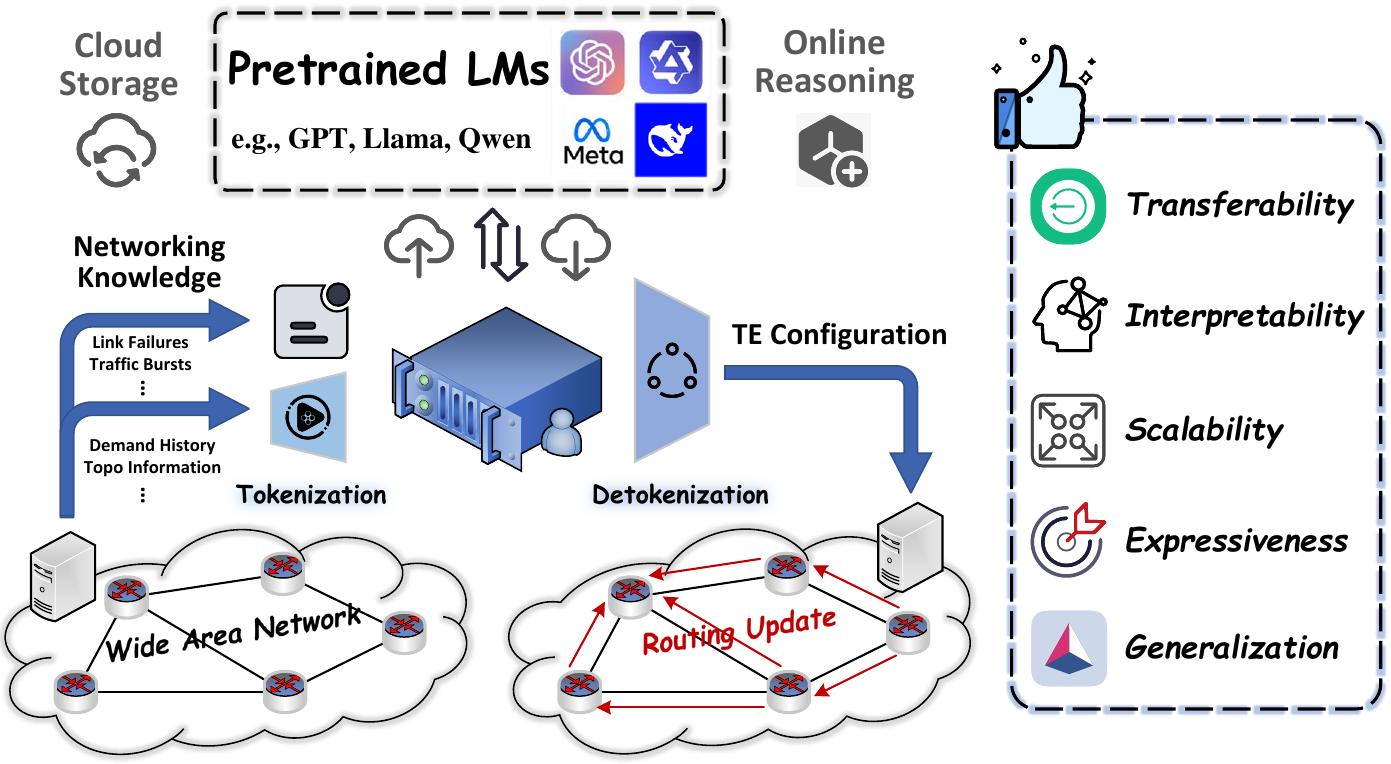}
\caption{Our ideal LM-driven TE system design for software-defined WANs.}
\label{fig:lm_in_cloud}
\end{center}
\end{wrapfigure}

After decades of research, there are still two major issues for conventional WAN TE relying on optimization tools like linear programming (LP): performance instability and computational complexity~\cite{abuzaid2021contracting,wang2006cope}. Recently, the growing presence of machine learning (ML) provides a second option: leveraging historical demands to quickly output a good routing scheme for future conditions~\cite{valadarsky2017learning}. These ML-based algorithms train deep neural networks (DNNs) using reinforcement learning (RL)~\cite{hope2021gddr,xu2023teal} or end-to-end supervised learning~\cite{perry2023dote,liu2024figret} to automatically infer WAN TE configurations, thereby eliminating the need for iterative optimization processes. 
Although promising, their imperfections can not be overlooked. As network environments evolve (e.g., due to topology changes or bursty traffic), the learned black-box mappings often fail to remain valid. This lack of generalization hinders the trustworthiness of their practical deployments. Consequently,  offline-trained DNNs soon become outdated under new conditions, necessitating frequent retraining. Besides, it is not uncommon to encounter caveats such as the following statement, which reflects a persistent yet unresolved concern over their models’ limited expressiveness:
\begin{center}
    ``$\mathsf{Our \ realization \ of \ DOTE \ uses \ a \ relatively \ simple \ NN. \cdots \ More \ expressive \ NN \ architectures}$ 
      $\mathsf{could  \ potentially \ lead \ to}$ $\mathsf{faster \ training \ and \ better \ quality \ solutions.}$” 
      $\quad \quad$- $\mathsf{Perry \ et \ al.}$~\cite{perry2023dote}
\end{center}
To address these shortcomings, we now turn to a brand-new paradigm that capitalizes on the latest breakthroughs in the ML community—namely, large language models (LMs) such as ChatGPT~\cite{chatgpt} and LLaMA~\cite{llama3modelcard}.

\textbf{LMs for WAN TE, Why?} These powerful models, with billions of parameters pre-trained on massive text corpora, have achieved remarkable \textit{generalization} capabilities in natural language processing (NLP). Encouragingly, recent work~\cite{wu2024netllm} has shown that LMs are \textit{transferable} to networking tasks, including viewport prediction, video streaming, and cluster scheduling. These successes suggest that LMs may be a good fit for WAN TE as well. To move beyond the intuition, we give a formal theoretical justification for their applicability to TE. Our core insight is that LMs solve complex problems in a manner akin to human reasoning~\cite{webb2023emergent, gonzalez2024does}, while TE falls into the category that benefits from such ability: it involves making sequential decisions, i.e., a solvable automaton (\S\ref{te_problem_setup}), despite its NP-hardness. This perspective lays the foundation for delving into the usage of LMs (\S\ref{lm_te_insight}). Specifically, we begin by investigating their \textit{expressiveness} when applied to state-to-state transitions within TE automata. As it turns out, an off-the-shelf LM, when exposed to sufficiently diverse data, can approximate any mapping from an arbitrary initial state to the near-optimal TE configuration.
Moreover, we provably show that LMs can simulate such mappings with only logarithmic depth. This capacity to abstract complex optimization procedures into conceptual leaps exemplifies \textit{reasoning}, further underscoring the suitability of LMs for WAN TE.

\begin{wrapfigure}[12]{r}{0.66\textwidth}
  \vspace{-0.5cm}
  \begin{center}
  \includegraphics[width=1.\linewidth]{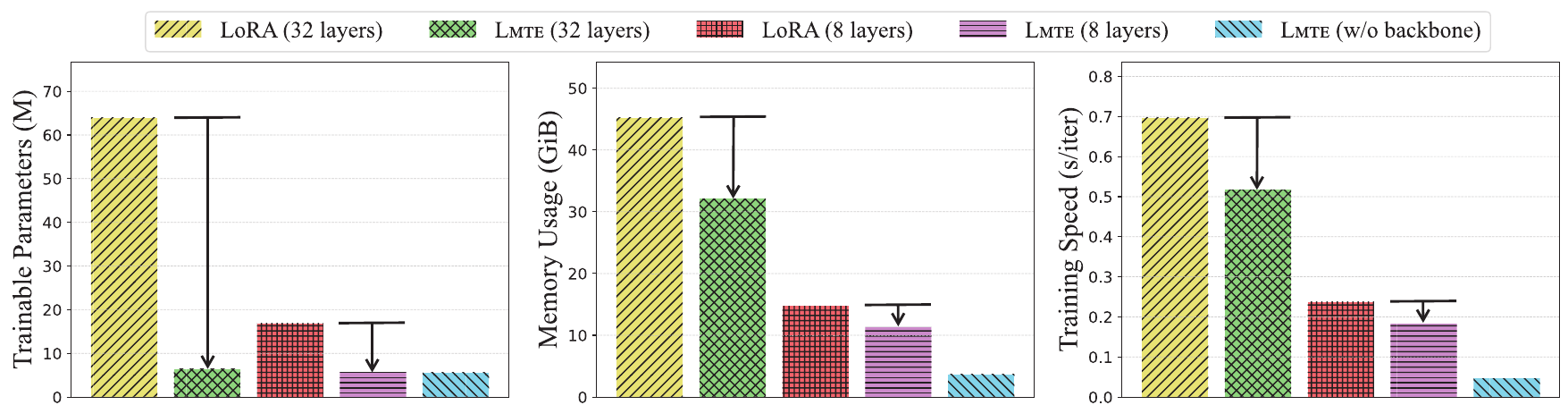}
\caption{Comparison of adaptation costs on a real-world topology: full-parameter fine-tuning of LLaMA-7B incurs high overhead, while \textsc{Lmte} is even more lightweight than LoRA by keeping the LM's backbone intact.}
\label{fig:lm_te_param}
\end{center}
\end{wrapfigure}

\textbf{LMs for WAN TE, How?} At first glance, our theoretical setup appears tractable for a vanilla LM. Unfortunately, the modalities of TE system inputs, i.e., historical traffic and topology information, differ significantly from the NL supported by LMs~\cite{wu2024netllm}. In addition, simply fine-tuning LMs is prohibitively expensive for TE applications—not only requiring substantial computational resources, but also gradually impairing their ability to generalize. To overcome these challenges, we introduce \textsc{Lmte} (\S\ref{lm_te_framework}), the \textit{first} framework designed to efficiently adapt LMs for TE optimization. \textsc{Lmte} enables LMs to understand both topology and traffic demands while retaining their pretrained knowledge. The key idea is to directly align inputs from multimodal encoders with prototype textual representations based on cross attention mechanism, progressively guiding the LM to interpret the task during adaptation (as visualized in \S\ref{subsec:visulization}). Further, we construct a domain-aware prompt template by preserving the LM's linguistic capabilities, to enhance \textit{interpretability} and \textit{robustness} under unseen cases. To support \textit{scalability}, \textsc{Lmte} allocates router-wise demands individually through a shared head network. As illustrated in Figure~\ref{fig:lm_in_cloud}, our envisioned LM-driven TE system takes as input high-level instructions and real-time data to support on-the-fly reasoning. Note that most computation is offloaded to the cloud, with the backbone model kept fixed. 
Meanwhile, Figure~\ref{fig:lm_te_param} shows \textsc{Lmte} achieves highly lightweight adaptation, e.g., only around 1\% local trainable parameters in the 32-layer LLaMA-7B model, even compared to methods like LoRA~\cite{wu2024netllm}.

We conduct a comprehensive evaluation of \textsc{Lmte} across diverse settings (\S\ref{exp_section}), including three real-world networks with publicly available traffic matrices and two large-scale topologies with synthetic data. Experimental results demonstrate that \textsc{Lmte} consistently outperforms state-of-the-art (SOTA) TE algorithms, particularly ML-based methods, achieving on average a 10\% reduction in maximum link utilization (MLU). By leveraging the LM’s pretrained knowledge and reasoning ability, \textsc{Lmte} generalizes well to a range of challenging and previously unseen scenarios. For instance, in comparison with \textsc{FIGRET}~\cite{liu2024figret}, a recently proposed robustness-enhanced method, \textsc{Lmte} surpasses it by up to 19\%, 14\%, and 21\% in performance under network failures, abrupt traffic, and distributional drift, respectively, on G\'{E}ANT dataset.
Furthermore, \textsc{Lmte} also offers efficiency gains—achieving over $10 \sim 100 \times$ faster runtime than LP solvers on two large topologies. 

\vspace{-0.2cm}
\paragraph{Our contributions} In summary, we make the following contributions:
\begin{itemize}
    \item We provide a theoretical analysis of the applicability of LM–based traffic engineering through a simplified automaton model, shedding light on how their reasoning capability may enable improved generalization in downstream tasks.
    \item We put forward \textsc{Lmte} as the first LM-driven traffic engineering framework for WANs, combining domain knowledge with data-driven planning through large language models.
    \item We extensively validate our proposal on various WAN traffic matrix datasets with pre-trained language models (i.e., the LLaMA family), achieving consistent performance gains over relevant baselines.
\end{itemize}

\section{Related Works \& Preliminaries}

\paragraph{Classical TE}
(WAN) traffic engineering is a well-known and established topic whose fundamental goal is to control and manage networks efficiently, playing a critical role in service and cloud provider infrastructures~\cite{jain2013b4,abuzaid2021contracting,narayanan2021solving,wang2006cope,kumar2018semi,applegate2003making,perry2023dote,xu2023teal,wang1999explicit,alqiam2024transferable,kumar2015bwe,racke2002minimizing,azar2003optimal}. Traditional TE solutions make routing decisions relying on linear programming techniques to achieve the optimal or sub-optimal performance~\cite{narayanan2021solving,abuzaid2021contracting,kumar2015bwe,wang2006cope,racke2002minimizing,applegate2003making,azar2003optimal}. Most existing algorithms fall into a category we refer to as prediction-based TE~\cite{jain2013b4,narayanan2021solving,abuzaid2021contracting}. These approaches collect a set of sampled traffic matrices and compute routing plans based on them. As a result, when actual traffic deviates significantly from the samples, the computed routing may perform poorly. In contrast to adaptive TE, oblivious routing seeks to optimize worst-case performance guarantees across all possible demand matrices~\cite{racke2002minimizing,applegate2003making,azar2003optimal}, thereby offering robust but highly sub-optimal performance for normal cases. 
An improved work is COPE~\cite{wang2006cope}, which achieves MLU optimization within the space spanned by historical TMs without compromising worst-case performance bounds. Unfortunately, solving the associated mathematical programs is still computationally expensive, as the number of variables grows rapidly with network size. 
Although NCFlow~\cite{abuzaid2021contracting} and POP~\cite{narayanan2021solving} significantly speed up computations via problem decomposition and parallelization, their methods inherit the limitations of demand prediction.

\paragraph{ML-Driven TE}
Recent advances in deep learning have spurred extensive research into ML–driven solutions for traffic engineering problems~\cite{valadarsky2017learning,xu2018experience,zhang2020cfr,hope2021gddr,bernardez2021machine,xu2023teal,perry2023dote,liu2024figret}. Not only is inference using neural models much faster than traditional models but also they offer the potential for autonomously managing
complex networked systems. The first class is RL-based methods~\cite{valadarsky2017learning,xu2018experience,zhang2020cfr,hope2021gddr,bernardez2021machine,xu2023teal}, typically using demand-prediction and RL approaches. While widely adopted in present-day literature, this family of solutions typically incurs high training overhead and exhibits sensitivity to parameter choices. The second class employs end-to-end optimization~\cite{perry2023dote,liu2024figret,alqiam2024transferable}, constituting a training strategy closely aligned with ours. Following the approach of DOTE~\cite{perry2023dote}, this class of schemes trains a simple decision model using historical traffic demands to directly produce SOTA configurations. Most recently, FIGRET~\cite{liu2024figret} further enhanced the robustness by incorporating a regularization term. While they achieve runtimes several orders of magnitude faster than solving a linear program, (unlike \textsc{Lmte}) whether their DNNs can serve as generalizable decision-makers needs further investigation. Concurrent with this work, \textsc{Pram}~\cite{yuan2026divide} establishes the state of the art for solving multi-commodity flow problems using multimodal language models. In contrast, our work explores a different design point in which only language models are used to support reasoning for WAN traffic engineering, without relying on additional input modalities (e.g., images).


\paragraph{Automata Theory}
A (deterministic) \textit{automaton} is formally defined as a triple $\mathcal{A} := \left(\Sigma, \mathcal{Q}, \delta\right)$, where $\Sigma$ represents a finite input alphabet (the set of allowed symbols), $\mathcal{Q}$ denotes a finite set of states, and $\delta: \mathcal{Q} \times \Sigma \rightarrow \mathcal{Q}$ is the total transition function that uniquely determines the next state for each current state and input symbol combination. $\mathcal{A}$ becomes an acceptor when equipped with an initial state $q_0 \in \mathcal{Q}$ and a set of accepting states. For any positive integer $T$ and $q_0$, the automaton induces a mapping from input sequences $(\sigma_1, \dots, \sigma_T) \in \Sigma^T$ to state sequences $(q_1, \dots, q_T) \in \mathcal{Q}^T$ defined by the recurrence relation $q_t := \delta
\left(q_{t-1}, \sigma_t\right), \forall t \in \left[1, \dots, T\right]$.

\paragraph{LMs Architecture}
In this paper, we adopt the prevailing Transformer-based architecture for LMs, using a decoder-only structure. An \textit{l}-layer Transformer~\cite{vaswani2017attention} is a sequence-to-sequence network, consisting of alternating self-attention blocks and feedforward MLP blocks. Concretely, we consider:
\[
    {\text{LM}} = f_{c} \circ f_{mlp}^{(l)} \circ f_{attn}^{(l)} \circ \cdots \circ f_{mlp}^{(1)} \circ f_{attn}^{(1)} \circ {\text{PE}},
\]
where $f_{c}: \mathbb{R}^{T \times D} \times \Theta_{c} \rightarrow \mathbb{R}^{T \times D}$ denotes classifier output layer and ${\text{PE}}$ is positional embedding. The self-attention layer first projects an input $\boldsymbol{H} \in \mathbb{R}^{T \times D}$ into three subspaces, namely $\boldsymbol{Q} \in \mathbb{R}^{T \times D_q}$, $\boldsymbol{K} \in \mathbb{R}^{T \times D_k}$ and $\boldsymbol{V} \in \mathbb{R}^{T \times D_v}$. Then the output is calculated as \( \boldsymbol{H}' = {\text{\scshape{softmax}}} \left( \frac{\boldsymbol{Q}\boldsymbol{K}^\top}{\sqrt{D_k}} \right)\boldsymbol{V} \). Specifically, it calculates the dot product between each token pair after projection, and the softmax is applied row-wise.

\section{TE Problem Setup: An Automata Perspective}
\label{te_problem_setup}

To lay the groundwork for the discussions that follow, we begin with a standard formulation of the TE system, outlining its basic components and objectives. We then recast the iterative optimization problem through an automata-theoretic lens, offering a structured perspective on its sequential dynamics and decision-making processes.

\subsection{Problem Formulation}

\begin{wrapfigure}[15]{r}{0.58\textwidth}
\vspace{-0.6cm}
\centering
\subfigure[Toy network topology\label{fig:te_problem_1}]{
    \includegraphics[width=0.3\textwidth]{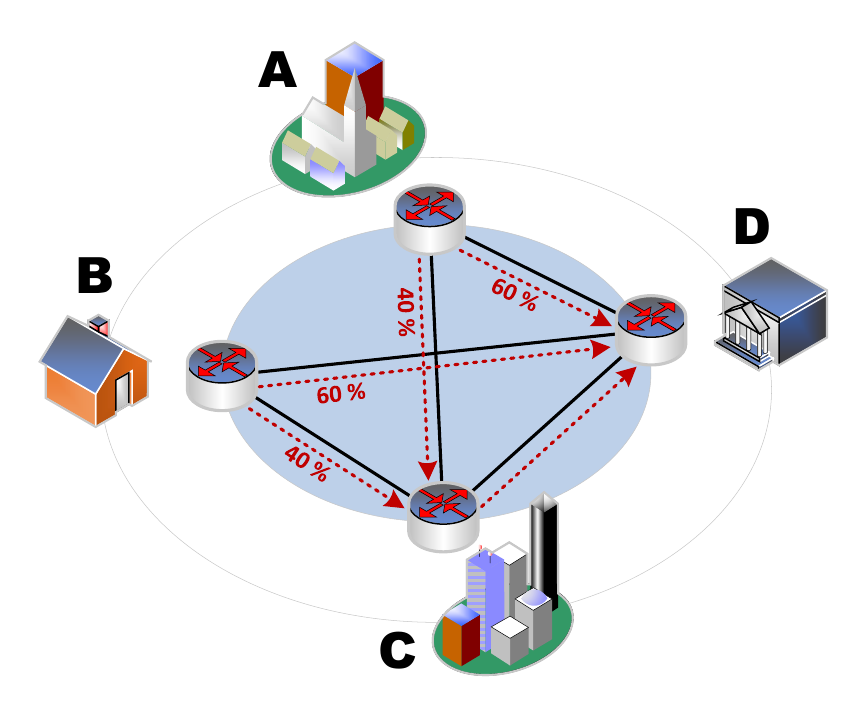}
}
\hfill
\subfigure[MLU function of the splitting\label{fig:te_problem_2}]{
    \includegraphics[width=0.25\textwidth]{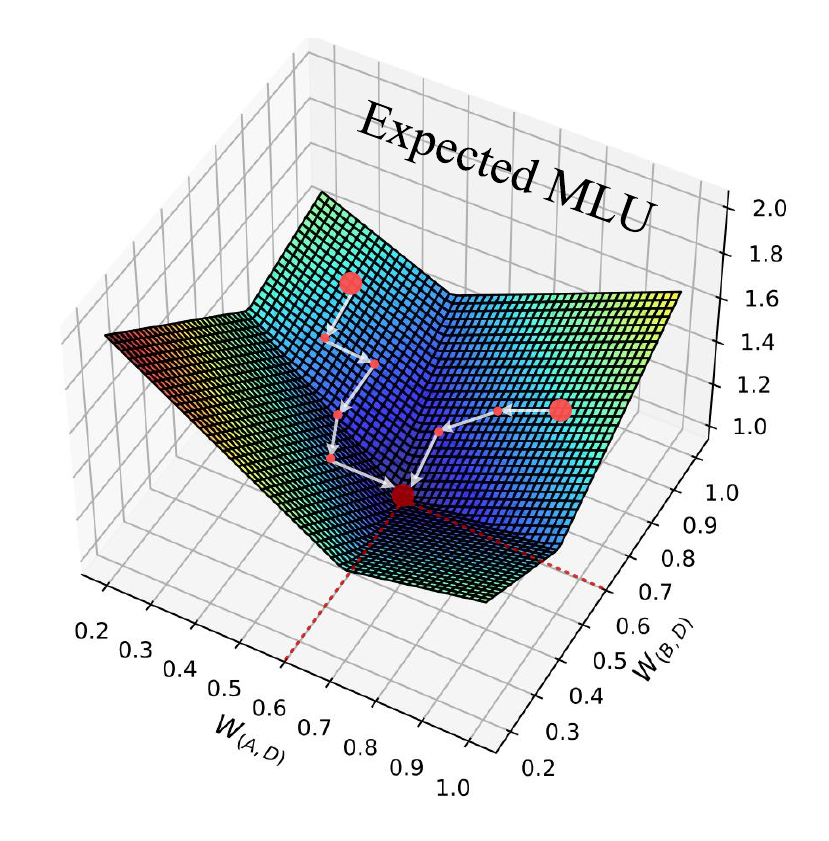}
}
\caption{Illustrative network traffic engineering (TE) example\textsuperscript{\cite{perry2023dote}}.}
\label{fig:te_problem}
\end{wrapfigure}

We model the network as a directed graph $\mathcal{G}(\mathcal{V},\mathcal{E},c)$ where vertices $\mathcal{V}$ represent network nodes (routers or switches), edges $\mathcal{E}$ correspond to directed communication links, and $c:\mathcal{E}\to\mathbb{R}^+$ assigns capacity to each link. Each origin-destination (OD) pair $(s,t)$ communicates via \textit{pre-selected} path sets $P_{s,t}$, called \textit{tunnels}. The \textit{traffic matrix} (TM) $\mathcal{D} = [\mathcal{D}_{i,j}] \in \mathbb{R}_{\geq 0}^{|\mathcal{V}|\times|\mathcal{V}|}$ provides a complete demand specification, where $\mathcal{D}_{i,j}$ represents the traffic demand of OD pair $(i,j)$ per time interval. Now given measured TMs, a \textit{traffic engineering configuration} $\mathcal{R}$ determines the mapping $\mathcal{D}_{s,t} \rightarrow \{r_p\}_{p\in P_{s,t}}$ where $r_p$ is the split ratio on tunnel $p$ for next time interval. Note that the split ratios should satisfy \( \sum_{p\in P_{s,t}} r_p = 1.0 \). 

This paper focuses on maximum link utilization (MLU)~\cite{azar2003optimal,kandula2005walking,perry2023dote}, a classical TE objective, which refers to the maximum value of all link utilization ratios in the network. A corresponding TE configuration can be obtained by solving the following mathematical program:
\begin{empheq}[box=\fbox]{equation}
    \begin{aligned}
       \textbf{minimize } & \quad \alpha = \max_{e \in \mathcal{E}} \frac{f_e}{c_e} \\ 
       \textbf{subject to} & \quad f_e = \sum_{s,t \in V} \sum_{p \in P_{s,t}} \sum_{e \ni p} \mathcal{D}_{s,t} \cdot r_p, \quad \forall e \in \mathcal{E};
        \quad r_p \ge 0, \quad \forall s,t \in \mathcal{V}, \ \forall p \in P_{s,t}; \\ 
        & \quad \sum_{p \in P_{s, t}} r_p = 1.0, \quad \forall s,t \in \mathcal{V}. 
    \end{aligned}
    \label{eqn:mlu_formulte}
\end{empheq}

We here illustrate the case using a example network~\cite{perry2023dote}, as depicted in Figure~\ref{fig:te_problem_1}, where all link capacities are set to $1$. At fixed time intervals, the TE system determines traffic splitting ratios for origin nodes \textit{A} and \textit{B}, distributing their demands across two tunnels to destination \textit{D}. The demands of \textit{A} and \textit{B} are independently drawn at each interval from a fixed distribution: with probability $1/2$, \textit{A}'s demand is $5/3$ and \textit{B}'s is $5/6$; otherwise, \textit{A} demands $5/6$ and \textit{B} demands $5/3$. Detailed in~\cite{perry2023dote}, the optimal split ratios for (\textit{A}, \textit{D}) and (\textit{B}, \textit{D}) are achieved at the red point ($0.6$, $0.6$) marked in Figure~\ref{fig:te_problem_2}.
Importantly, this subfigure further shows that the (expected) MLU is \textit{\textbf{convex}} with respect to the split ratios—a desirable property that is not only \textit{commonly assumed} in TE optimization problems~\cite{wang2006cope,perry2023dote}, but also forms the cornerstone of the problem transformation we introduce next.

\subsection{Tailored Problem}

In particular, the TE function in Figure~\ref{fig:te_problem} suggests a discrete optimization trajectory: starting from arbitrary split ratios, the configuration is iteratively adjusted along the curve of (expected) MLU until convergence to a global minimum. This step-by-step adjustment naturally aligns with the concept of a decision chain where each action incrementally improves the TE outcome.
Motivated by the observations, we can model TE configuration selection as an automaton $\mathcal{A}$. The state space $\mathcal{Q}$ is defined as all possible solutions to the TE problem, i.e., the collection of split ratios for the communication tunnels. The input alphabet $\Sigma$ comprises two components: candidate OD pairs and their updated tunnel weights. Let $q_0 \in \mathcal{Q}$ be the WAN condition-dependent initial state. As illustrated in Figure~\ref{fig:automata}, a simple $\mathcal{A}$ processes an input $\sigma_t \in \Sigma$ and transitions to state $q_{t+1}$ by optimizing the selected pair(s)’s allocation at each step $t$, e.g., from (1, 6) to (4, 6). The automaton halts at state $q_T$ once the improvement in MLU falls below a predefined threshold, thereby yielding the final TE configuration. The objective of TE is thus to compute $q_T$, the terminal state of $\mathcal{A}$ after $T$ transitions. We then prove it can be modeled in \textit{finite-length}:
\begin{lemma}
    (TE as Solvable Automata) WAN traffic engineering is NP-hard, but any instance admits a solution through finite-step automaton simulation. 
    \label{lemma:1}
\end{lemma}
The proof of Lemma~\ref{lemma:1} can be found in Appendix~\ref{app:lem1}. The first result shows that our automaton-based formulation is generally solvable. However, due to the NP-hardness of the underlying mathematical program, solving it via sequential optimization becomes computationally prohibitive at scale, as illustrated at the top of Figure~\ref{fig:te_automata_problem}. This underscores the need for decision models with sufficient expressive capacity to support practical deployment.

\begin{figure}
    \subfigure[An automata-style iterative optimization of the TE objective\label{fig:automata}]{
        \centering
        \includegraphics[width=0.705\linewidth]{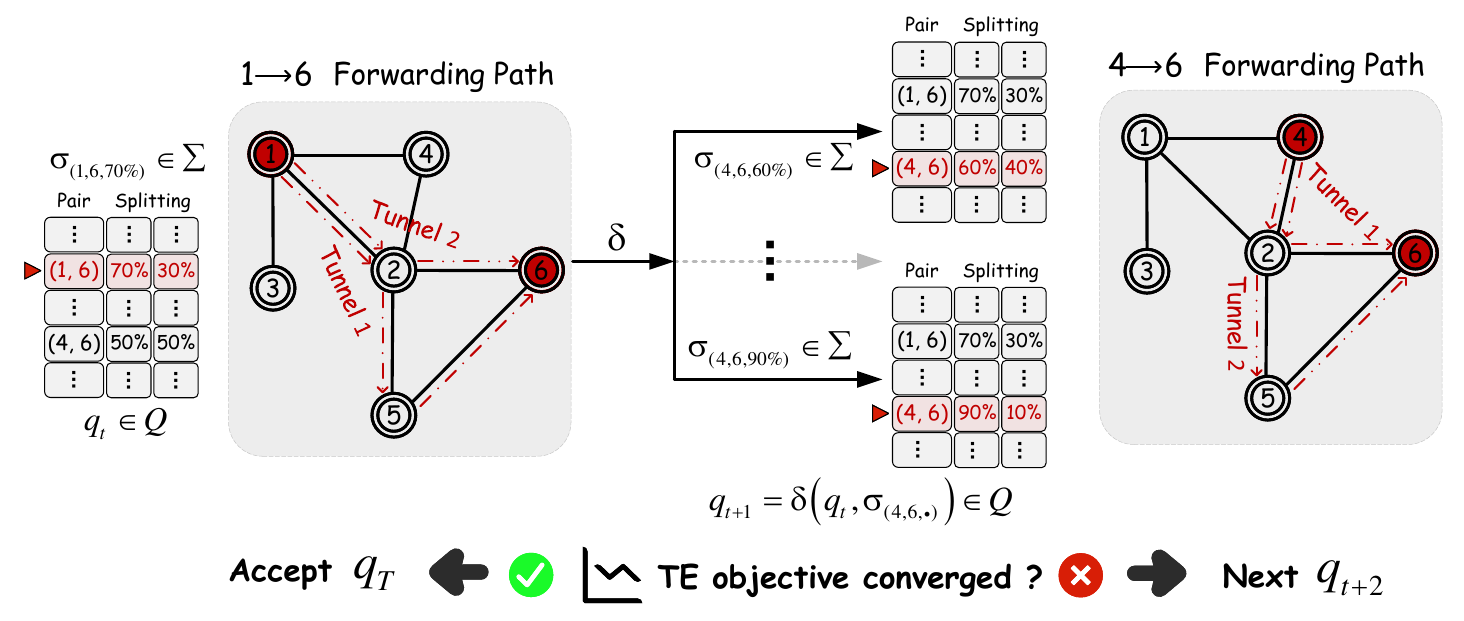}
    }
    \hfill
    \subfigure[TE automata simulation\label{fig:te_automata_problem}]{
        \centering
        \includegraphics[width=0.265\linewidth]{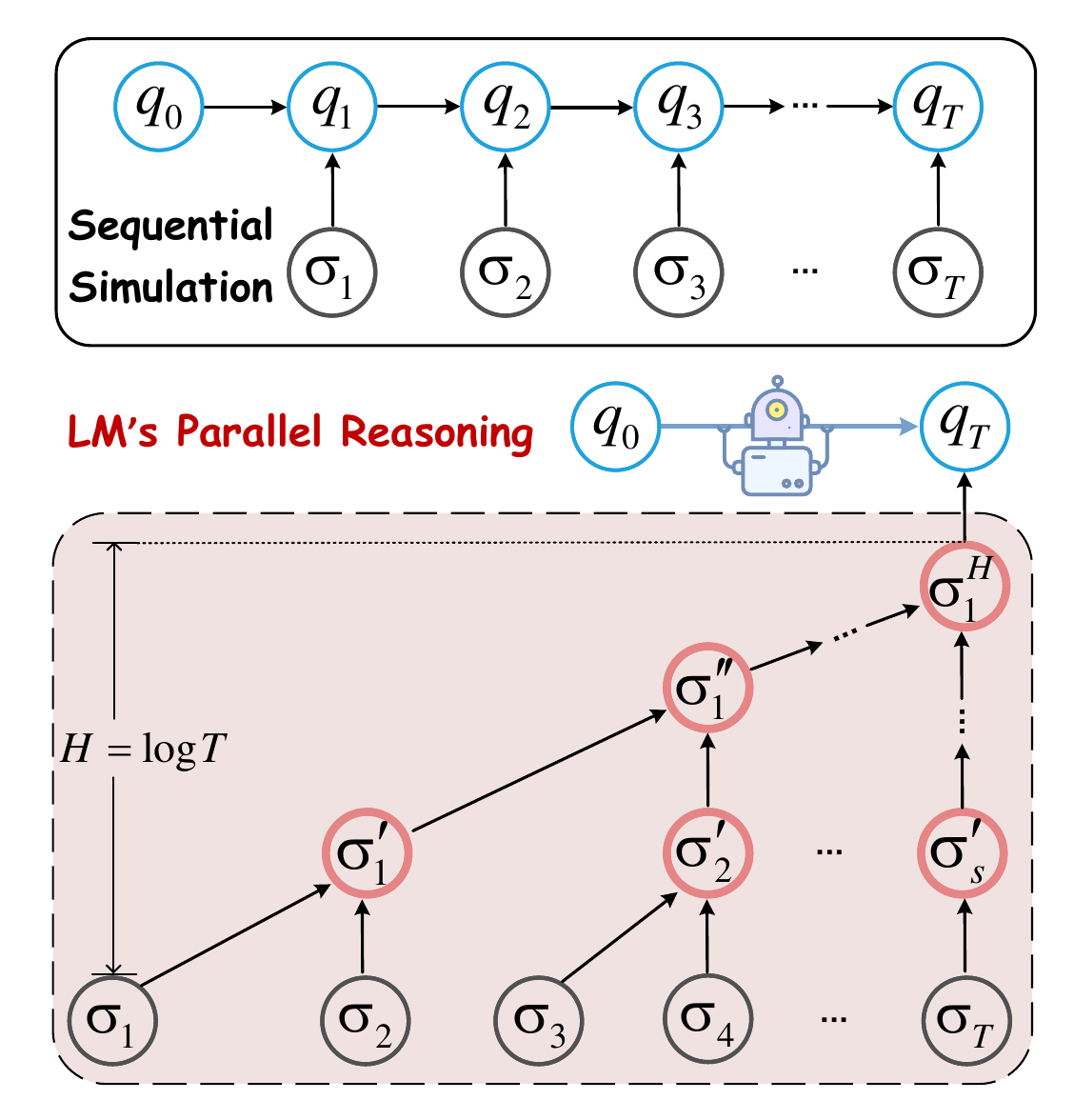}
    }
    \caption{Theoretical intuition and foundations for our LM-driven TE approach. \textbf{(a)} Simple TE example from a solvable automaton perspective (Lemma~\ref{lemma:1}). \textbf{(b)} A pretrained LM can approximate any optimal state from an initial state (Theorem~\ref{theorem:2}), and accelerate automaton simulation with logarithmic-depth (Theorem~\ref{theorem:3}).}
\end{figure}

\section{Harnessing LMs for TE: Theoretical Insights}
\label{lm_te_insight}

We next investigate pre-trained language models, known for their strong generalization and emergent reasoning capabilities~\cite{webb2023emergent}, as a promising decision model to TE. This section provides theoretical insights into their applicability to our automaton-based formulation.

\subsection{LMs Have Enough Expressiveness}

We start by posing a fundamental question: can LMs discover and reach the optimal state in WAN TE? This hinges on whether LMs are expressive enough to capture and simulate the underlying state transitions. To explore this, we consider the TE state $q_i$ of the automaton as compositional sequence vector like $\underbrace{u_1 \ v_1 \ c[u_1][v_1], \dots, d[u_1][v_1], \dots}_{\text{input conditions \& historical demands (prefix)}}$, $\underbrace{r_1[u_1][v_1], \ldots}_{\text{split ratios}}$, analogous to the natural language token. With this representation, we now state the following theorem:
\begin{theorem}
   (LMs as Powerful Simulators) For any \(\epsilon > 0\) and state-to-state function \(f\), we can always find an LM \(h\) satisfying \( d(f(\mathbf{x}), h(\mathbf{x})) \le \epsilon \), where $d$ denotes a discrepancy over states, typically induced by a function norm $\|f - h\|_p$. 
   \label{theorem:2}
\end{theorem}
The proof of Theorem~\ref{theorem:2} can be found in Appendix~\ref{app:the2}. Theorem~\ref{theorem:2} implies that, in terms of expressiveness, LMs can simulate general TE automata from $q_0$ to $q_T$, by implicitly encoding $\delta$, $\sigma_{1:T}$ within their parameters, given convergence training on corpora of sufficient richness. 
Nevertheless, as indicated by~\cite{alqiam2024transferable}, trivially deep architecture (e.g., recurrent adjustments with $\mathcal{O}(T)$-depth and appropriate nonlinearities) may still be required to obtain the near-optimal TE configuration. 
This raises the question of whether LMs can leverage their ``reasoning'' power to circumvent the inefficiency.

\subsection{LMs Bring ``Reasoning'' to TE}

To this end, we highlight the parallel computation capabilities of LMs, by extending results from prior works~\cite{sanford2024transformers,liu2023transformers}.
Specifically, any parallelizable TE automaton can be simulated through a tree-structured process (Figure~\ref{fig:te_automata_problem}, bottom), enabling all states to be computed within just $\mathcal{O}(\log T)$ layers. This resolves the earlier question and demonstrates how LMs can perform non-sequential, compositional inference—hallmarks of \textit{reasoning}.
Theorem~\ref{theorem:3} formalizes this key insight, showing that LMs can solve the TE problem with a high probability under logarithmic depth. Taken together with Theorem~\ref{theorem:2}, these results establish the fundamental soundness of LMs for TE.
\begin{theorem}
    (LMs Find TE Shortcuts) An LM with feedforward MLP dimension $D=\mathcal{O}\left(|\mathcal{Q}|^2\right)$, attention dimension $D_{(\cdot)}=\mathcal{O}\left(|\mathcal{Q}|\right)$ and depth logarithmic in $T$ can perform continuous simulation of any TE automaton $\mathcal{A} = \left(\Sigma, \mathcal{Q}, \delta\right)$ at length $T$ with probability at least $1 - 2/T^2$. 
    \label{theorem:3}
\end{theorem}
The proof of Theorem~\ref{theorem:3} can be concluded from~\cite{liu2023transformers}, which is presented in Appendix~\ref{app:the3}. It shows that transformer-based LMs can learn shortcuts to TE models significantly more complex than deterministic finite automata. Formally, this means that given a sequence $q_0, q_1, q_2, \dots, q_T$ as states, the pre-trained LM simulates it by $q_0 \rightarrow q_1, q_2, q_4, \dots, q_T$.

\section{LM-Driven TE Framework: \textsc{Lmte}}
\label{lm_te_framework}

\begin{figure}
    \centering
    \includegraphics[width=1.\linewidth]{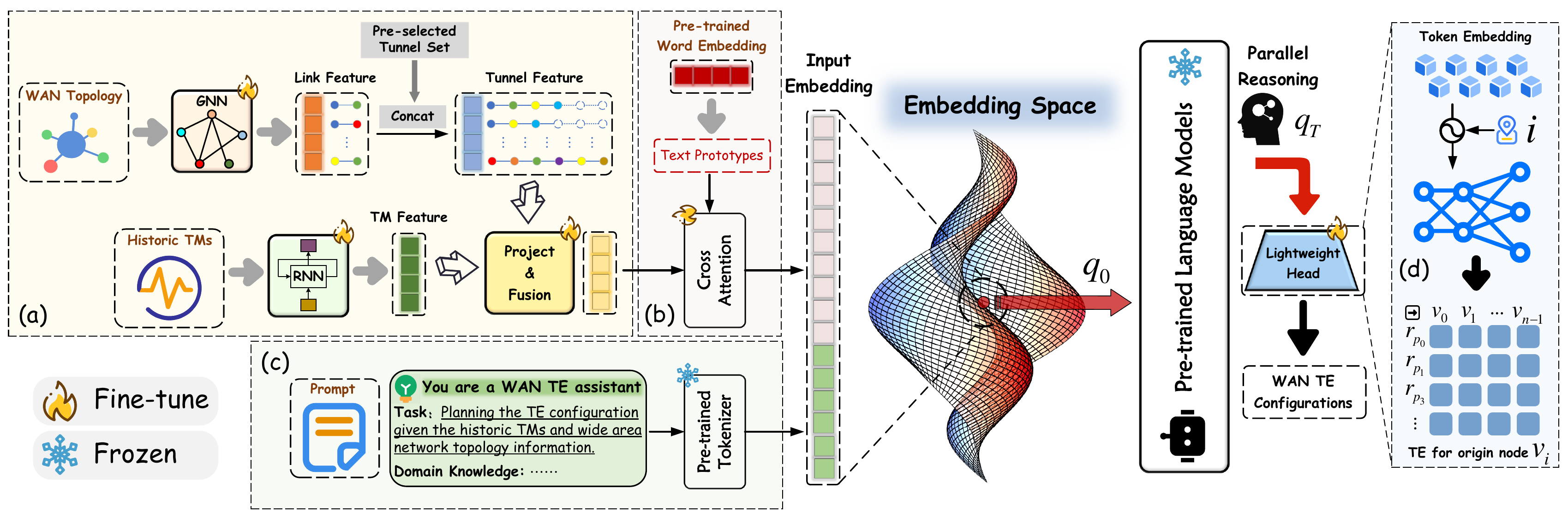}
    \caption{An overview of the proposed \textsc{LMTE}, the first LM–driven WAN TE framework. \textbf{Left:} input embedding alignment to bridge the modality gap. \textbf{Middle:} a frozen LM backbone, reused without any modification. \textbf{Right:} a lightweight LM head that jointly outputs TE configurations.}
    \label{fig:overview}
\end{figure}

The previous section identifies LMs as compelling candidates for the decision-making tasks of interest. However, certain aspects of the theoretical analysis are difficult to implement directly within LMs, e.g., constructions of the (latent) TE automata. The central challenge is how to bridge the modality gap between lengthy symbolic specifications and pretrained LM inputs and outputs. While adaptation is a simplest fix, it can be resource-intensive, potentially at the expense of generalization, as mentioned in \S\ref{sec:introduction}.
Below, we elaborate on the design of \textsc{Lmte}—a \textbf{L}anguage \textbf{M}odel-driven \textbf{T}raffic \textbf{E}ngineering framework—developed to address the dilemma, including cross-modal alignment (\S\ref{comp1}, \S\ref{comp2}), prompting support (\S\ref{comp3}), and efficient adaptation (\S\ref{comp4}). 

\subsection{Invariant Spatial-Temporal Embedding}\label{comp1} 

\textsc{Lmte} first encodes the multimodal inputs of TE problems into token-like embeddings, using a GNN encoder for network topology and an RNN encoder for historical TMs. As our formulation (\S\ref{te_problem_setup}) requires preconfigured tunnel information as an input, we explicitly form these paths using features generated by the GNN encoder~\cite{alqiam2024transferable}. Specifically, the node features are aggregated over each link and concatenated with the corresponding capacity to model edge features. For each tunnel, the features of all traversed edges are then concatenated to construct a tunnel-level representation.
For example, the embedding of tunnel 1$\rightarrow$2$\rightarrow$4$\rightarrow$3 is $\left[ E^{1,2}; E^{2,4}; E^{4,3}\right]$, where $E^{i,j}$ denotes the edge feature connecting nodes \#$i$ and \#$j$. In doing so, \textsc{Lmte} remains robust under various WAN changes (e.g., links, tunnels, or capacities). We refer to the property as \textit{state invariance}, where $\text{embedding}\left(\boldsymbol{T} \circ \boldsymbol{X} \right) = \boldsymbol{T} \circ \text{embedding}\left(\boldsymbol{X} \right)$ for an arbitrary transformation $\boldsymbol{T}$ and WAN state $\boldsymbol{X}$. Figure~\ref{fig:overview}\textcolor{purple}{(a)} summarizes the embedding pipeline.
To unify the structural and temporal representations, we then reproject the temporal embeddings $\boldsymbol{R}^{\mathcal{D}} \in \mathbb{R}^{S \times D}$ using a learned matrix $\mathbf{W}$ derived from the tunnel embeddings $\boldsymbol{R}^{\mathcal{T}} \in \mathbb{R}^{P \times L}$, where $S$ denotes the historical window size and $P$ the number of tunnels. The resulting embedding is computed as
\[
\boldsymbol{R}^{\mathcal{F}}_{s,c} = \sum_{d} \boldsymbol{R}^{\mathcal{D}}_{s,d} \cdot \mathbf{W}_{d,c}, \ \text{with } \mathbf{W} = \phi\left(f_{\theta}\left(\left[ \boldsymbol{R}^\mathcal{T}_1; \dots; \boldsymbol{R}^\mathcal{T}_P \right]\right)\right).
\]
Here, $f_\theta(\cdot)$ denotes a neural transformation (e.g., an MLP) applied to the concatenated tunnel features, and $\phi(\cdot)$ reshapes the output into $\mathbf{W}$. This mechanism enables the input model to align temporal dynamics with structural context by conditioning the projection on the underlying topology. An additional advantage of this design is that, since the output dimension $C$ is tunable via $f_\theta$, the fused representation $\boldsymbol{R}^{\mathcal{F}} \in \mathbb{R}^{S \times C}$ is much more compact than a naïve stack of the two modalities.

\begin{figure}
    \centering
    \includegraphics[width=0.9\linewidth]{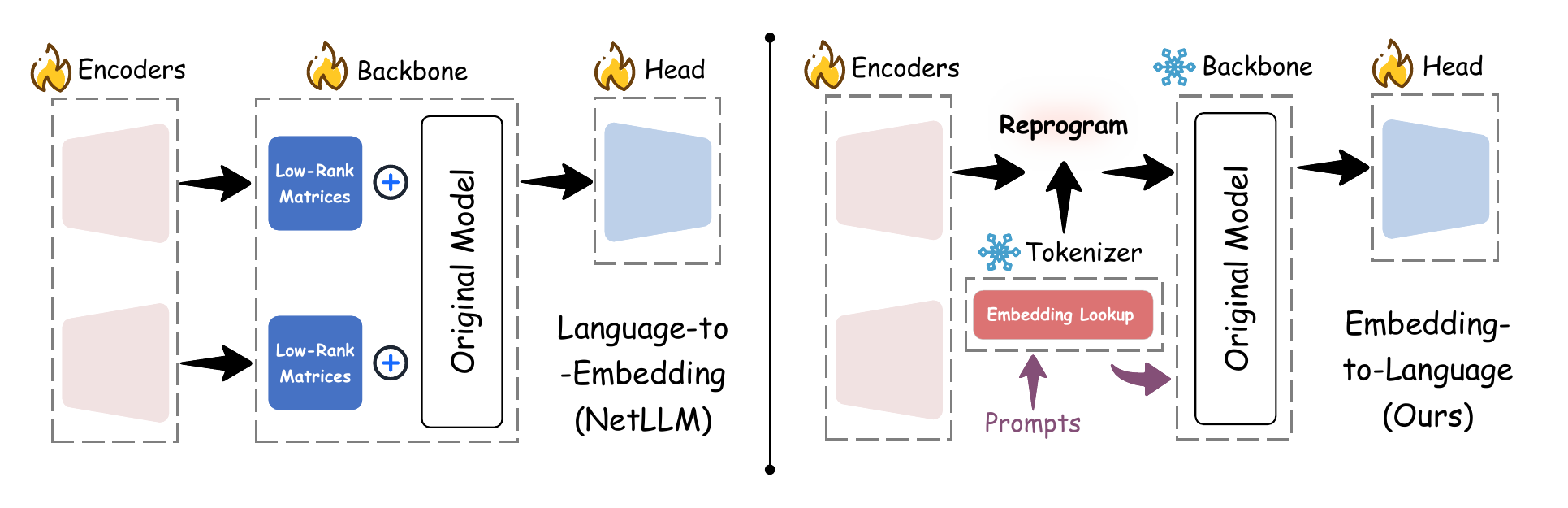}
    \caption{Illustration of our embedding-to-language alignment in comparison of language-to-embedding alignment (e.g., LoRA used by \textsc{NetLLM}~\cite{wu2024netllm}).}
    \label{fig:lmte_lora}
\end{figure}

\subsection{Embedding-to-Language Alignment}\label{comp2}  

\begin{wrapfigure}[14]{r}{0.65\textwidth}
  \vspace{-0.3cm}
  \begin{center}
  \includegraphics[width=0.9\linewidth]{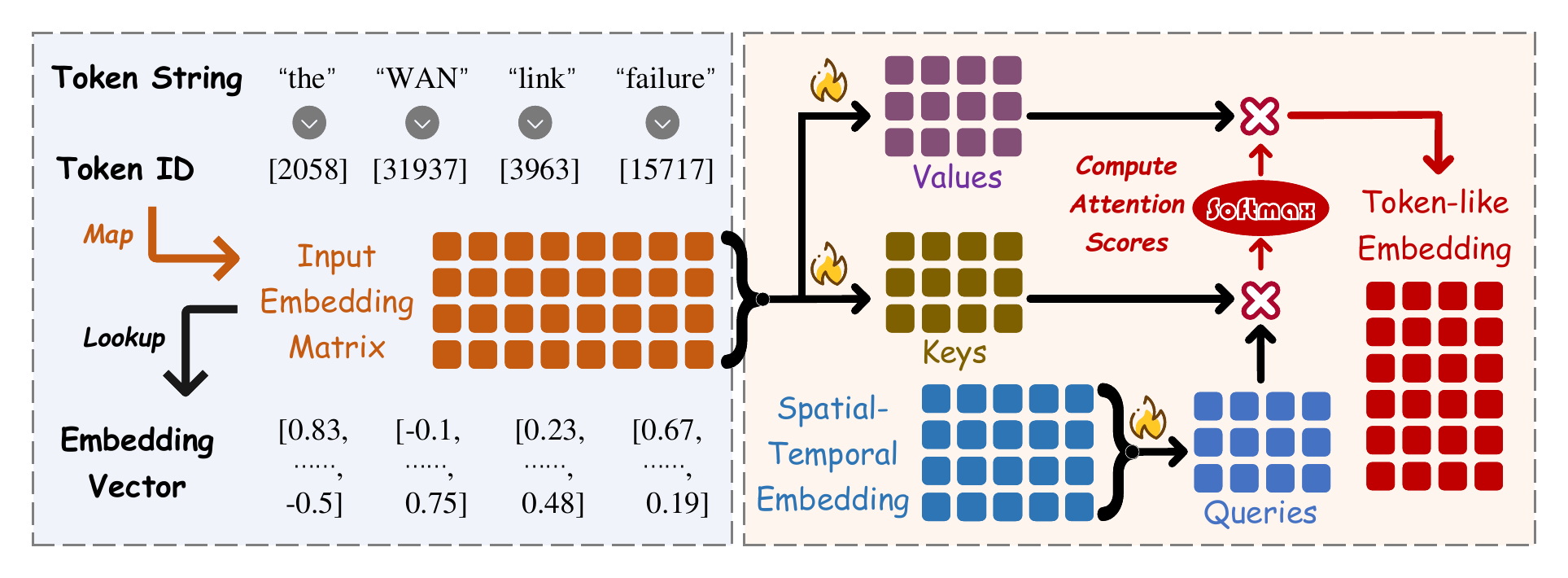}
\caption{Reprogramming spatial-temporal representations using original input embeddings of LMs. \textbf{(left)} The tokenizer maps and lookups token embedding vectors. \textbf{(right)} We obtain token-like embeddings via cross attention.}
\label{fig:reprogramme}
\end{center}
\end{wrapfigure}

Now, the framework must align tokenized representations understood by LMs with the input ones. A common language-to-embedding approach is LoRA~\cite{wu2024netllm}, which injects trainable low-rank matrices into selected layers of the backbone. Compared to full-parameter fine-tuning, this strategy is technically feasible for bridging heterogeneous data modalities in task-specific settings. However, as it still modifies the backbone parameters, it may gradually forget previously acquired knowledge, diminishing generalization—an issue known as catastrophic forgetting~\cite{wang2023orthogonal}. Instead, \textsc{Lmte} adopts an embedding-to-language approach, leaving the backbone unchanged while offering even greater cost-efficiency (see Figure~\ref{fig:lm_te_param}). The core idea is to reprogram the spatial-temporal embeddings into the source NL space~\cite{jin2024time}, as illustrated in Figure~\ref{fig:lmte_lora}. Before delving into the details, we first clarify how input tokens are embedded in an LM: each token is mapped to a unique identifier, which is then used to retrieve a vector from a learned input embedding matrix. The pipeline is summarized on the left of Figure~\ref{fig:reprogramme}. Since this matrix inherently serves as a set of \textit{text prototypes} that ground the LM’s understanding, our alignment mechanism then reduces to a linear lookup over these pre-trained word embeddings, guided by the WAN information. 

As shown on the right of Figure~\ref{fig:reprogramme}, \textsc{Lmte} employs a cross attention to realize the embedding-to-language alignment. Concretely, $\boldsymbol{R}^{\mathcal{F}}$ is linearly projected to the query $\mathbf{W}_k^Q \boldsymbol{R}^{\mathcal{F}} \in \mathbb{R}^{n_q \times d_k}$, for a attention head $k$. The same operation applies to the key $\mathbf{W}_k^K \boldsymbol{R}^{\mathcal{N}} \in \mathbb{R}^{n_k \times d_k}$ and value $\mathbf{W}_k^V \boldsymbol{R}^{\mathcal{N}} \in \mathbb{R}^{n_k \times d_v}$, where $\boldsymbol{R}^{\mathcal{N}}$ denotes the result of a linear projection to the input embedding matrix. We then calculate the cross-attention as 
\[
  \boldsymbol{R}^{\mathcal{A}}_k = {\text{\scshape{softmax}}}\left(\frac{\left(\mathbf{W}_k^Q \boldsymbol{R}^{\mathcal{F}}\right)\left(\mathbf{W}_k^K\boldsymbol{R}^{\mathcal{N}}\right)^\top}{\sqrt{d_k}}\right)\mathbf{W}^V_k\boldsymbol{R}^{\mathcal{N}}.
\]
When multi-head operation is applied, the outputs of each head are concatenated. The underlying rationale of the above procedure is to enable each spatial-temporal representation to query a shared bank of linguistic prototypes, so that \textsc{Lmte} can naturally seek interpretation of the task in the NL space.

\subsection{Task-Specific Prompting}\label{comp3} 

Another notable strength of \textsc{Lmte} over existing LM-based networking methods (i.e., \textsc{NetLLM}~\cite{wu2024netllm}) lies in its ability to maintain LMs' prompt input compatibility (see Figure~\ref{fig:lmte_lora}). This is particularly beneficial in our task-specific settings, where well-crafted prompts function as a vital interface between expert intent and model behavior. An example is shown in Figure~\ref{fig:overview}\textcolor{purple}{(c)}. On one hand, these prompts contribute to knowledge activation and task grounding. For instance, they can specify the LM’s role (e.g., ``$\mathsf{You \ are \ a \ WAN \ traffic \ engineer}$”), and include TE-specific context such as descriptions of the WAN or statistics of historical TMs. On the other hand, they provide a flexible mechanism for encoding ad hoc constraints, such as network changes or potential traffic pattern. Consider the case of link failures, by incorporating failure-related descriptions into the prompt (e.g., ``$\mathsf{Link \ 1–3 \ is \ currently \ down}$”), \textsc{Lmte} is able to adapt its output accordingly without retraining. 
The task-specific prompting thus not only improves interpretability but also enhances the robustness of \textsc{Lmte}, enabling flexible responses to unseen and dynamic conditions in volatile WAN environments. The overall prompt structure is shown in Figure~\ref{fig:prompt}.

\begin{figure}
    \centerline{\includegraphics[width=0.84\linewidth]{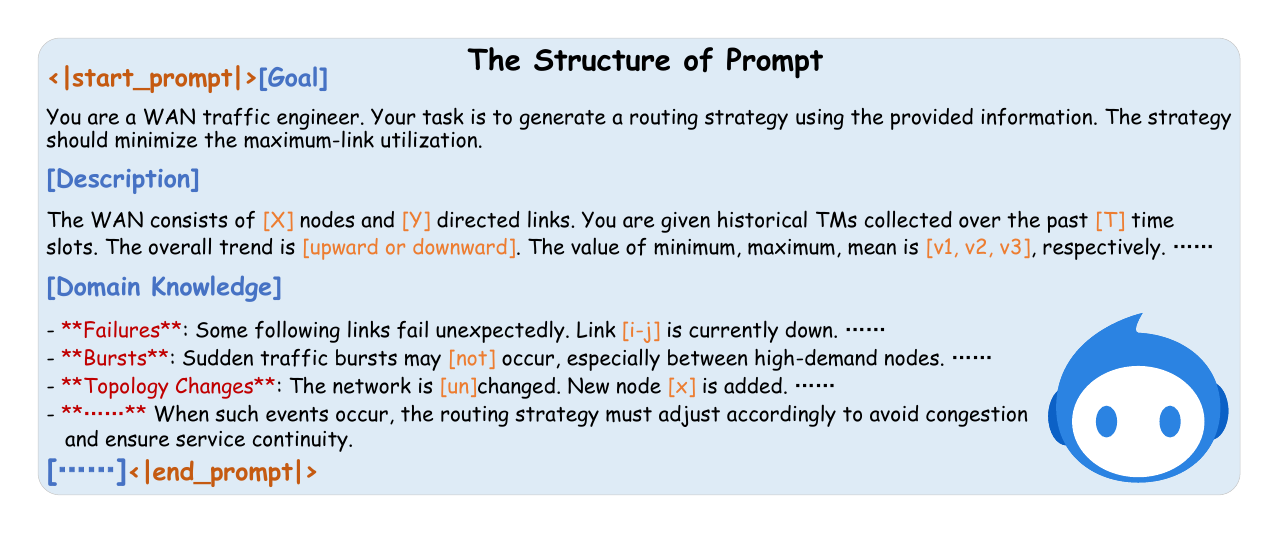}}
    \caption{The overall structure of the task-specific prompt used by \textsc{Lmte}.}
    \label{fig:prompt}
\end{figure}

\textbf{Why not include all tunnels directly in the prompt?} A WAN with over 120 nodes and 8 tunnels per OD pair yields $>$110K entries—exceeding LM's maximum input limits. Serializing this information directly into the prompt is therefore impractical.

\subsection{Lightweight End-to-End Adaptation}\label{comp4}  

In this part, we present how to obtain globally optimal allocations and fine-tune the model of \textsc{Lmte}. The first challenge here is \textsc{Lmte}'s scalability to keep pace with the growing size of the WAN. For a topology with thousands of nodes, the solution space can span millions of dimensions, resulting in an overwhelming number of trainable parameters. To break the ``curse of dimensionality'', inspired by multi-agent policy~\cite{xu2023teal}, \textsc{Lmte} generates split ratios independently via a shared, lightweight prediction head. As illustrated in Figure~\ref{fig:overview}\textcolor{purple}{(d)}, parameters are shared across origin nodes (or routers), which is specified to the projector using the Transformer sinusoidal positional embedding~\cite{vaswani2017attention}. This design reduces the neural network's computational complexity from $\mathcal{O}(\left|\mathcal{V}\right|^2)$ to $\mathcal{O}(\left|\mathcal{V}\right|)$, substantially lowering the memory overhead. Regarding model adaptation, we adopt an end-to-end direct optimization strategy. Similar to DOTE~\cite{perry2023dote}, the LM-driven system tries to infer fine-grained configurations for $\mathcal{D}^i_t$ from the recent sequence $\left\{ \mathcal{D}^i_{t-1}, \mathcal{D}^i_{t-2}, \dots, \mathcal{D}^i_{t-H-1} \right\}$, topology $\mathcal{G}(\mathcal{V},\mathcal{E},c)$, and its corresponding tunnels.
After computing the final softmax outputs for all OD pairs, the system evaluates the network-wide MLU loss in Eqn.~\ref{eqn:mlu_formulte}, whose gradients are computable for any past
realization~\cite{perry2023dote}. \textsc{Lmte} then updates its trainable parameters using the gradients of the expected MLU computed over each batch of data points. 
Once adapted, the alignment and reasoning of LMs can be done automatically. It is worth mentioning that only the lightweight encoding and head networks need to be deployed locally, while the \textit{frozen} backbone can be hosted remotely (as in Figure~\ref{fig:lm_in_cloud}) for resource-constrained environments.

\section{Evaluation}
\label{exp_section}

In this section, we begin by describing the experimental setup in \S\ref{exp_setup}, followed by a comparison of \textsc{Lmte} with SOTA TE methods in \S\ref{formal_ev}. We then evaluate its performance under various conditions (i.e., link failures, traffic bursts, and natural drift) in \S\ref{special_ev}, and its overhead in \S\ref{comlexity_ev}. Finally, additional experiments \& visualizations are presented in \S\ref{additional_ev}.

\subsection{Experimental Setup}\label{exp_setup}

\paragraph{Datasets}
We consider three real-world WAN topologies: US Internet2 Network~\cite{zhang2005network} (\textit{Abilene}, \textit{12} nodes, \textit{30} edges), China Education and Research Network~\cite{qiao2024autotomo} (\textit{CERNET}, \textit{14} nodes, \textit{32} edges) and Pan-European Research Network~\cite{uhlig2006providing} (\textit{G\'{E}ANT}, \textit{23} nodes, \textit{74} edges), each with traffic matrices captured at different snapshots. We also select two larger WAN topologies (\textit{Cogentco}, \textit{197} nodes, \textit{486} edges and \textit{UsCarrier}, \textit{158} nodes, \textit{378} edges) from the Internet Topology Zoo~\cite{knight2011internet} and generate \textit{non-stationary} synthetic TM sequences via a \textit{gravity model}~\cite{roughan2002experience}. Specifically, we model them by combining trend, sinusoidal (periodic) variations, and noise.
Note that we adopt a \textit{7:1:2} split ratio for training, validation, and testing. We choose \textit{12} as the history window size for the input TMs.

\paragraph{Baselines}
We compare \textsc{Lmte} against the following baselines: \circled{1} two state-of-the-art ML-driven approaches: \textit{DOTE}\cite{perry2023dote} and \textit{FIGRET}\cite{liu2024figret}, which make centralized decisions based on historical demand observations; and \circled{2} three Optimization-based approaches: \textit{COPE}\cite{wang2006cope}, \textit{Oblivious Routing}\cite{azar2003optimal}, and \textit{Gurobi-Pred}, where Gurobi (v11.0.3)~\cite{gurobi} is employed to optimize predictions obtained by weighted moving average. 
Note that all the above baselines use the same set of candidate paths. Regarding the tunnel selection, we employed $4$ shortest paths for experiments on the two larger topologies. For others, we employed $8$ shortest paths as the default setting, consistent with previous research~\cite{perry2023dote}.

\paragraph{Implementing \textsc{Lmte}}
We implement \textsc{Lmte} using PyTorch and PyTorch Geometric, and conduct experiments on a $16$-core virtual machine equipped with NVIDIA RTX 4090 24GB GPUs. Unless otherwise specified, we utilize LLaMA-3-8B with the first $8$ layers as the default backbone. For the RNN encoder, we adopt a bidirectional RNN architecture; for the GNN encoder, we adopt a Graph Transformer architecture. To improve the model’s semantic understanding during adaptation, we randomly introduce link failures
and traffic bursts with a probability of $10\%$. Additionally, the shared output projection is parallelized to accelerate the process by prioritizing tensor calculations over for loops. All experiments are repeated $3 \sim 5$ times, and we report the averaged results.

\subsection{TE Solution Quality}\label{formal_ev}

\begin{figure}
    \centerline{\includegraphics[width=1.\linewidth]{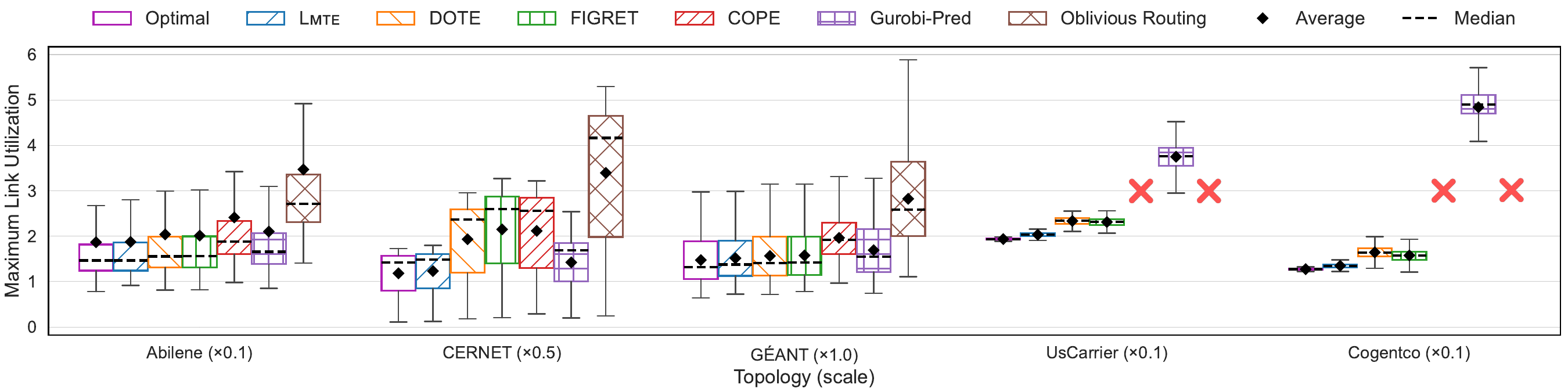}}
    \caption{TE quality of \textsc{Lmte} and baselines under the objective of minimizing maximum link utilization (MLU). The $x$-axis shows different topologies with their normalization scales, where MLU is adjusted by the scale factor. \textcolor{red}{Red crosses} indicate cases where no result is available due to computation failure.}
    \label{fig:te_performance}
\end{figure}

Figure~\ref{fig:te_performance} compares the normal-case quality of TE configurations between \textsc{Lmte} and other TE schemes on five topologies. The $y$-axis value represents the maximum link utilization. We observe that \textbf{\circled{1} ML-driven approaches:} The two DNN-based schemes, DOTE and FIGRET, are reaching near-optimal average performance among existing methods. However, their results exhibit a trade-off as the WAN scale increases. Furthermore, on the smaller CERNET topology, their behavior appears notably anomalous. Given the more dynamic network patterns in CERNET, it is plausible that these DNNs have merely overfitted to the historical data, thus indicating a lack of generalization ability. \textbf{\circled{2} Optimization-based approaches:} (Semi-)oblivious schemes inevitably incur high link utilization even under non-worst-case scenarios. More critically, algorithms like COPE and Oblivious Routing become impractical for large-scale topologies (e.g., those exceeding 120 nodes) due to their prohibitive memory requirements. While Gurobi-Pred performs well in environments with temporally stable traffic, its effectiveness degrades substantially in the presence of non-stationary demands. \textbf{\circled{3} LM-driven approach (\textsc{Lmte}):} Our proposed \textsc{Lmte}, grounded in language models, consistently achieves the lowest MLU across all topologies, outperforming both DNN- and optimization-based methods in average and median performance. For instance, compared to the best TE scheme on large and complex WANs, \textsc{Lmte} reduces the average MLU by $8\% \sim 15\%$. Benefiting from the pre-trained reasoning and abstraction capabilities of LMs, \textsc{Lmte} adapts effectively to diverse traffic patterns and topology scales without overfitting to specific network instances. These results underscore the potential of LMs as a powerful foundation for real-world, large-scale WAN TE.

\subsection{Performance under Unseen Scenarios}\label{special_ev}

\begin{figure}
    \subfigure[Reacting to random failures\label{fig:multiple_failures}]{
        \centering
        \includegraphics[width=0.39\linewidth]{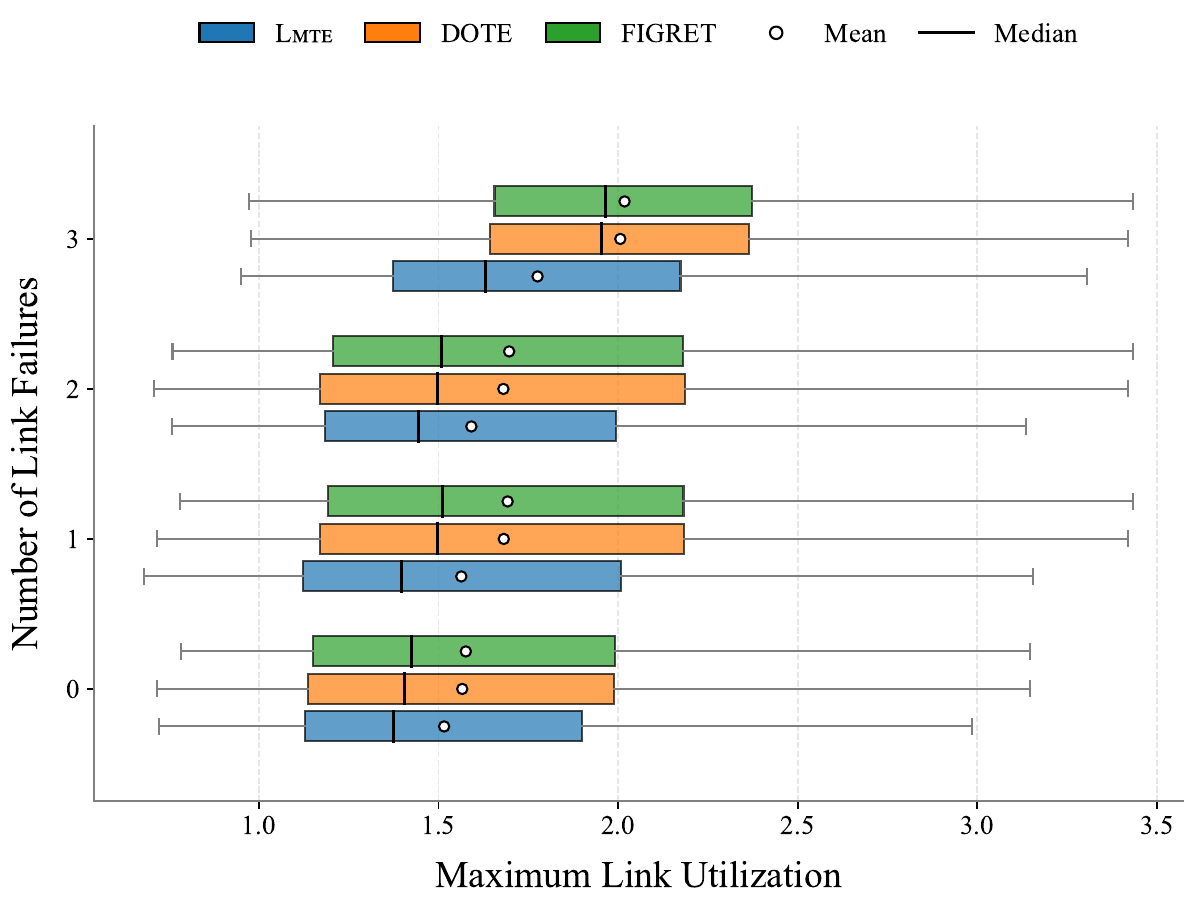}
    }
    \hfill
    \subfigure[MLU under different single failure scenarios\label{fig:single_failure}]{
        \centering
        \includegraphics[width=0.59\linewidth]{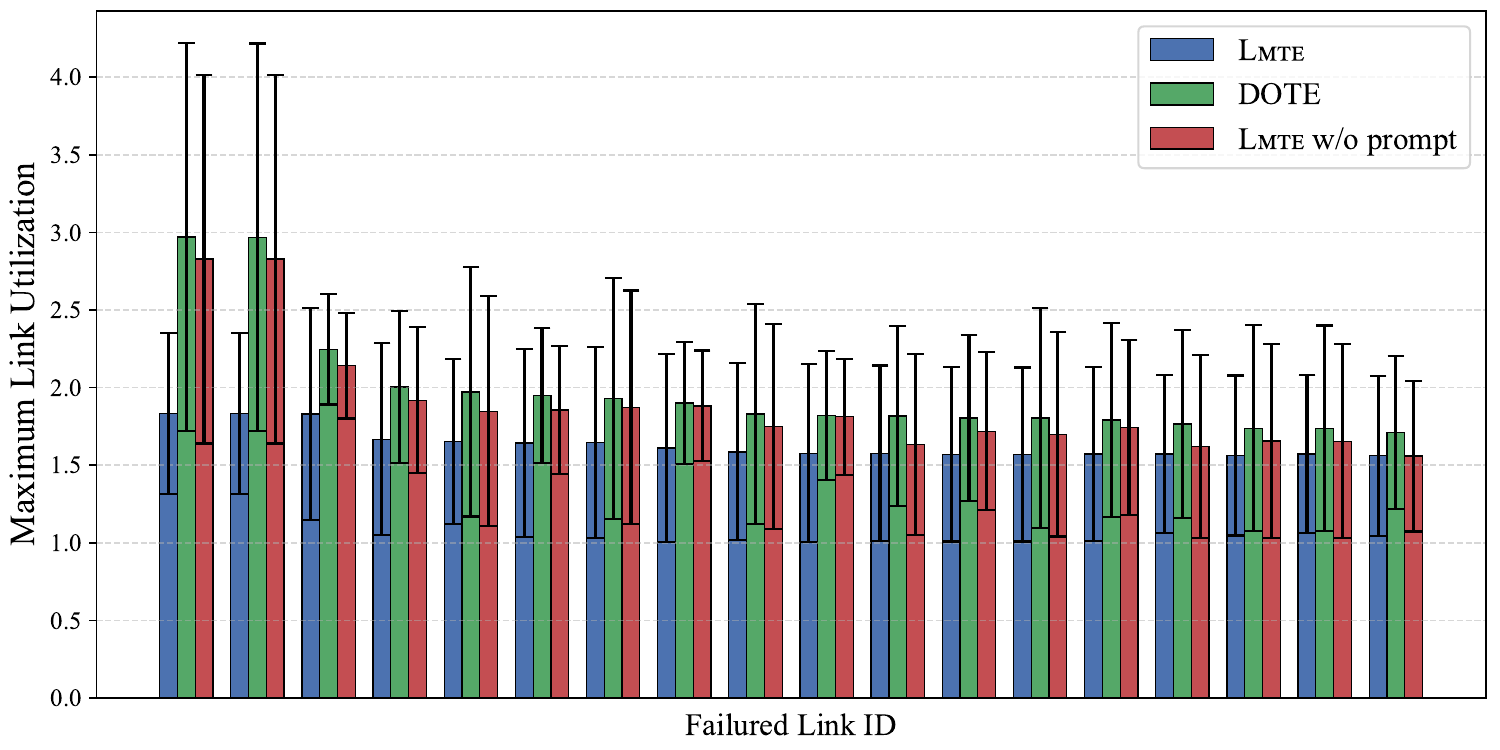}
    }
    \caption{Comparing the performance of \textsc{Lmte} and baselines on G\'{E}ANT under testing scenarios with link failures.}
\end{figure}

\paragraph{Coping with Network Failures}
Beyond prompt guidance, \textsc{Lmte} handles different link failures using simple heuristics that proportionally redistribute affected traffic across available paths, introducing minimal computational overhead. The same strategy is applied to the baselines.
Figure~\ref{fig:multiple_failures} compares \textsc{Lmte}, DOTE, and FIGRET under varying numbers of link failures on the G\'{E}ANT topology. While all methods maintain relatively stable performance, \textsc{Lmte} shows superior robustness, with the performance gap widening as failures increase.
To further investigate, we simulate single-link failures on the sorted $18$ most critical links. As shown in Figure~\ref{fig:single_failure}, \textsc{Lmte}'s average MLU remains between $1.6$ and $1.75$, closely matching its optimal performance in Figure~\ref{fig:te_performance}. The figure also highlights that prompt-guided \textsc{Lmte} performs especially well under the most severe failure scenarios.

\paragraph{Handling Natural Distribution Drift}
We test the impact of natural train-test distribution shifts on \textsc{Lmte} by fine-tuning with progressively earlier TM segments (0$\sim$25\%, 25$\sim$50\%, and 50$\sim$75\%), validating on 75$\sim$85\%, and testing on the final 15\%. This setup highlights the performance drop compared to training on the first 75\% of the TM. A detailed comparison with FIGRET is presented in Tab.~\ref{tab:generalization-drift}. \textsc{Lmte} shows strong resilience, with less than 10\% degradation even when only one-third of the traffic data is used for adaptation—except on the G\'{E}ANT dataset, where the data distribution is highly imbalanced. The increased volatility observed in FIGRET underscores the advantage of pretrained language models in few-shot adaptation and generalization.

\begin{table}[h]
\centering
\caption{Performance decline with natural distribution drift.}
\label{tab:generalization-drift}
\begin{tabular}{c|c|ccc}
\toprule
\multirow{2}{*}{\textbf{Topology}} & \multirow{2}{*}{\textbf{Scheme}} & \multicolumn{3}{c}{\textbf{Training traffic matrices time segments}} \\
\cmidrule(lr){3-5}
 & & \textbf{0\%$\sim$25\%} & \textbf{25\%$\sim$50\%} & \textbf{50\%$\sim$75\%} \\
\midrule
\multirow{2}{*}{Abilene} 
  & \textsc{Lmte} & $6.9 \% \pm 0.50 \%$  & $4.7 \% \pm 0.46 \%$ & $3.8 \% \pm 0.81 \%$ \\
  & FIGRET        & $10.5 \% \pm 0.53 \%$  & $7.6 \% \pm 0.49 \%$ & $8.5 \% \pm 0.76 \%$ \\
\cmidrule(lr){1-5}
\multirow{2}{*}{CERNET} 
  & \textsc{Lmte} & $7.8 \% \pm 0.84 \%$  & $6.5 \% \pm 0.56 \%$ & $3.2 \% \pm 0.51 \%$ \\
  & FIGRET        & $10.6 \% \pm 1.08 \%$  & $9.1 \% \pm 0.57 \%$ & $16.7 \% \pm 1.55 \%$ \\
\cmidrule(lr){1-5}
\multirow{2}{*}{G\'{E}ANT} 
  & \textsc{Lmte} & $15.3 \% \pm 1.36 \%$  & $16.7 \% \pm 1.60 \%$ & $2.5 \% \pm 0.18 \%$ \\
  & FIGRET        & $20.8 \% \pm 0.61 \%$  & $30.6 \% \pm 0.73 \%$ & $5.8 \% \pm 0.34 \%$ \\
\bottomrule
\end{tabular}
\end{table}

\begin{wrapfigure}[13]{r}{0.45\textwidth}
  \vspace{-0.25cm}
  \begin{center}
  \includegraphics[width=1.\linewidth]{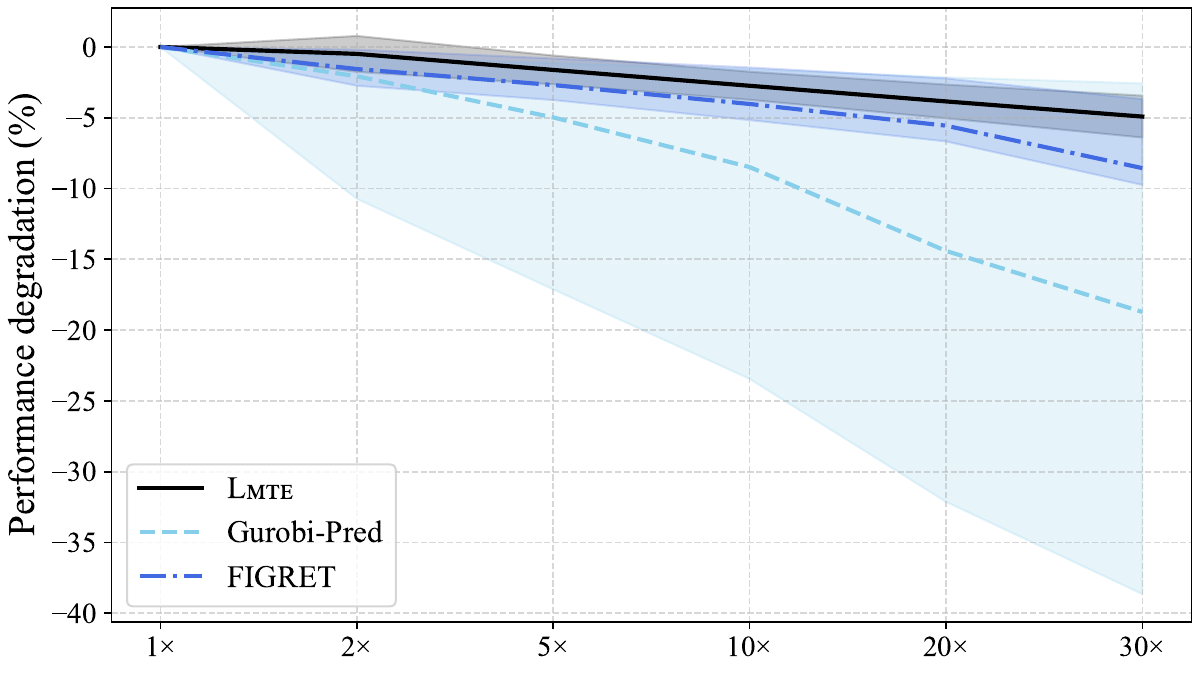}
\caption{Performance degradation with fluctuations.}
\label{fig:burst_degration}
\end{center}
\end{wrapfigure}

\paragraph{Robustness to Demand Changes}
To evaluate \textsc{Lmte}’s robustness under drastically changing demand, we introduce temporal fluctuations to the test TMs, by following a similar approach to~\cite{xu2023teal}. For each OD-pair, we compute the variance of its demand over time and scale it by factors of $2$, $5$, $10$, $20$, and $30$ to define zero-mean normal distributions. Samples drawn from these distributions are then added to each demand at every time slot. Results are shown in Figure~\ref{fig:burst_degration}. Compared to FIGRET, which is specifically designed to handle such scenarios, \textsc{Lmte} achieves comparable or even greater stability. The performance does not decrease by more than $5 \%$. In contrast, prediction-based TE schemes suffer a significant degradation in MLU relative to the other two ML-based methods.

\subsection{Computation Complexity}\label{comlexity_ev}

\begin{wrapfigure}[15]{r}{0.75\textwidth}
\vspace{-0.25cm}
\centering
\subfigure[Computation Time vs. WAN Scale\label{fig:computation_cost}]{
        \centering
        \includegraphics[width=0.495\linewidth]{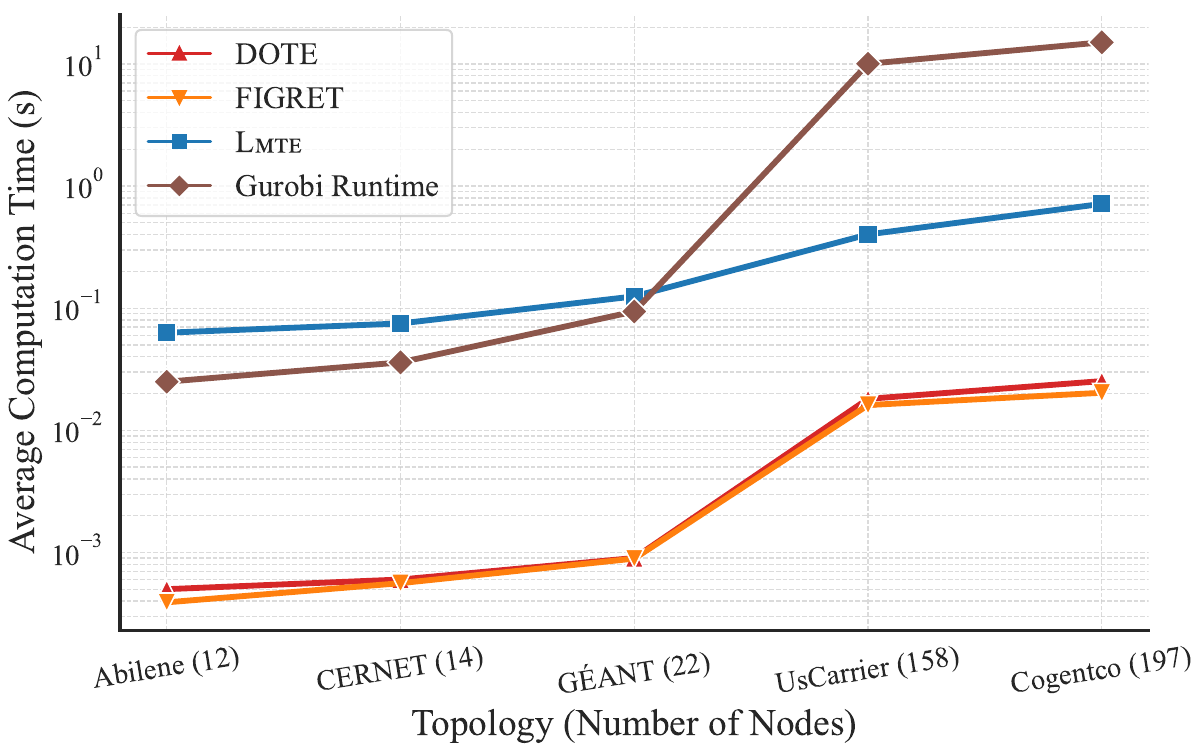}
    }
    \hfill
    \subfigure[Model Size vs. WAN Scale\label{fig:share_params_cost}]{
        \centering
        \includegraphics[width=0.47\linewidth]{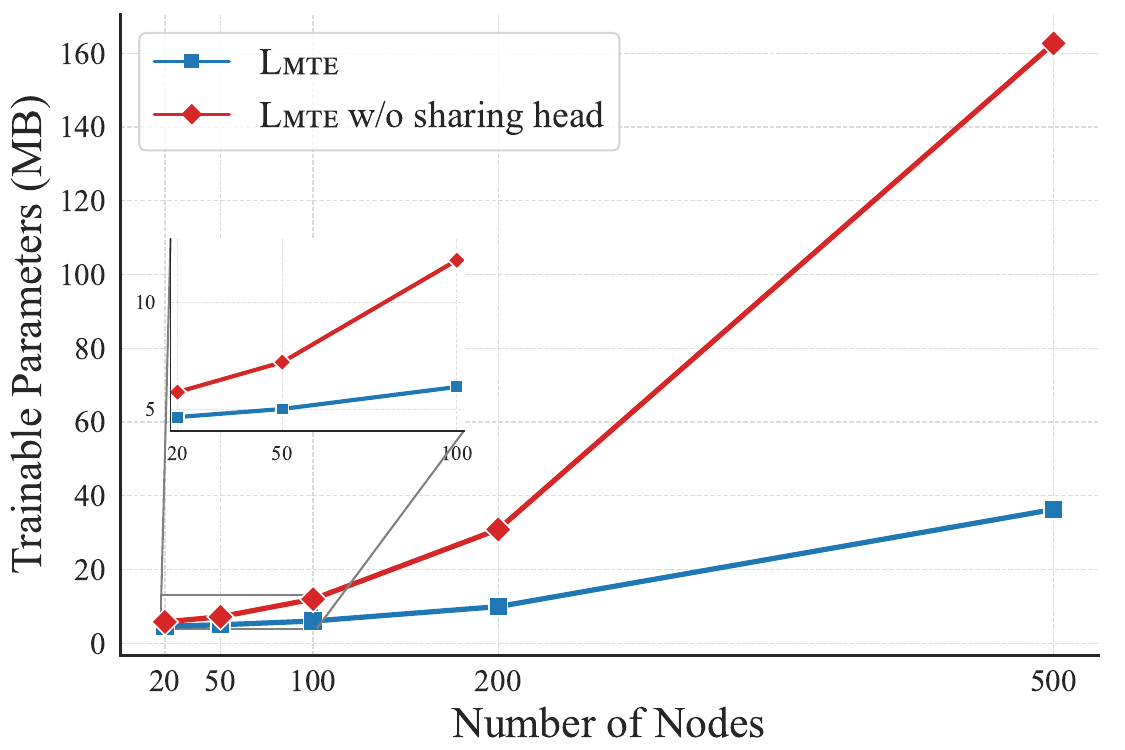}
    }
\caption{Comparison of computation complexity across different WAN scales.}
\end{wrapfigure}

Figure~\ref{fig:computation_cost} compares the runtime of \textsc{Lmte} against other ML-based methods and Gurobi. DOTE and FIGRET, both relying on lightweight MLP architectures, achieve faster inference than \textsc{Lmte}. However, \textsc{Lmte} still delivers over an order-of-magnitude speedup compared to Gurobi on large topologies, suggesting its runtime remains practical for real-world deployment. For instance, loading a 7B LM like LLaMA-2-7B takes about 0.1s $\sim$ 0.3s to generate one answer, and it takes about 0.04s for smaller LMs such as OPT-1.3B to generate one answer~\cite{wu2024netllm}. This provides ample time for measurement and follow-up actions within each TE cycle. 
We also report the total number of tunable parameters\footnote{Half-precision is used for parameter computation in our experiments.} 
in Figure~\ref{fig:share_params_cost}. As the number of nodes increases, the memory required by \textsc{Lmte} decreases substantially, owing to our shared LM head design.

\subsection{Additional Results \& Discussions}\label{additional_ev}

\paragraph{Impact of LMs' Parameter Scales}
To assess the applicability of \textsc{Lmte} across different model scales, we evaluate it on three additional LMs beyond 3-8B: 3.2-1B, 3.2-3B, and 2-13B. A control variant without any backbone is also included to examine the LM's standalone reasoning capability. As shown in Figure~\ref{fig:models}, larger models generally yield better performance, with diminishing returns beyond 8B. Moreover, all adapted LMs outperform the control variant on both topologies, demonstrating the effectiveness of leveraging LMs. Meanwhile, as shown in the figure, a 3B-parameter model achieves over 90\% of the performance of its 13B-parameter counterpart. Therefore, \textsc{Lmte} achieves a favorable performance–cost tradeoff, indicating that strong traffic engineering performance can be attained without resorting to very large language models, which is appealing for practical deployment.

\begin{figure}
\centering
    \subfigure[Performance v.s. backbones\label{fig:models}]{
        \centering
        \includegraphics[width=0.47\linewidth]{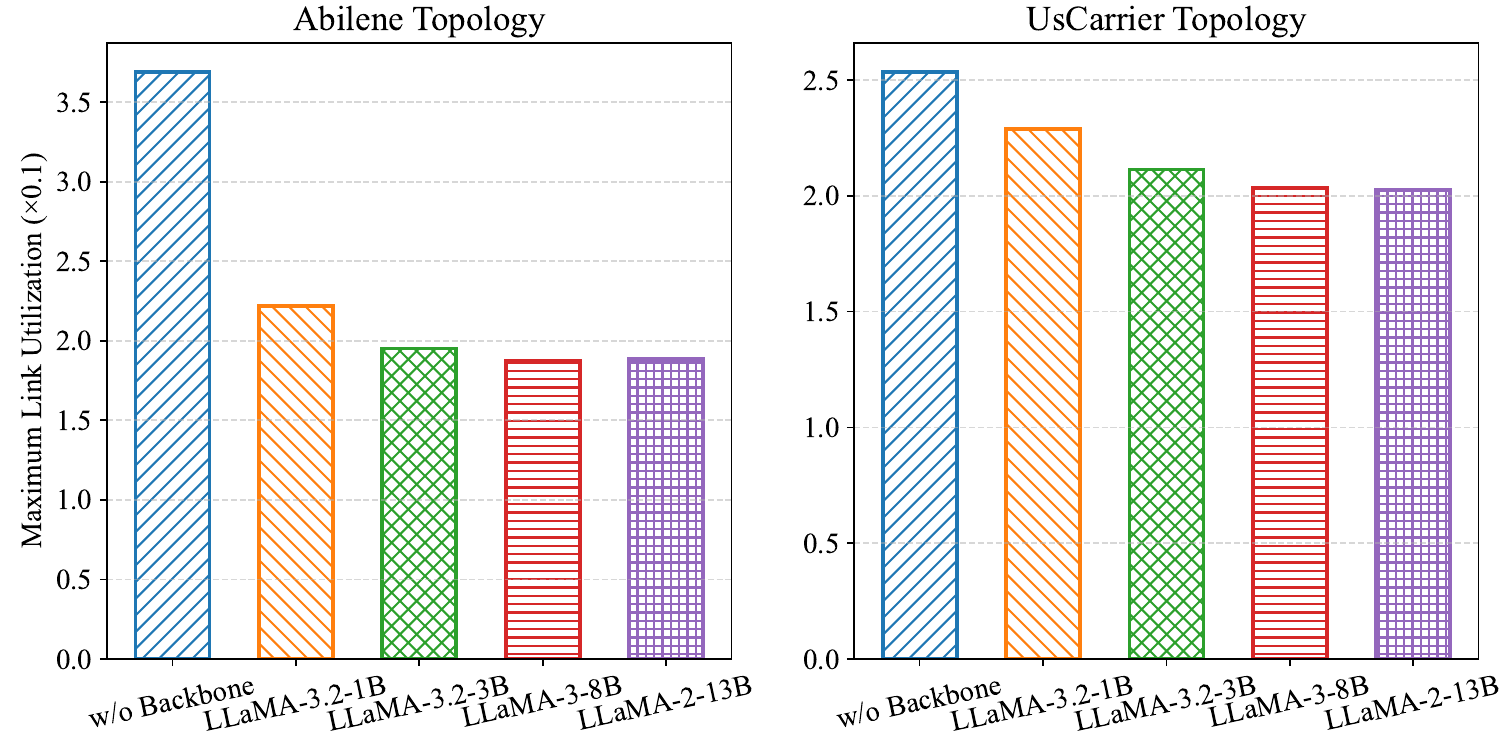}
    }
    \subfigure[Performance v.s. layers\label{fig:layers}]{
        \centering
        \includegraphics[width=0.43\linewidth]{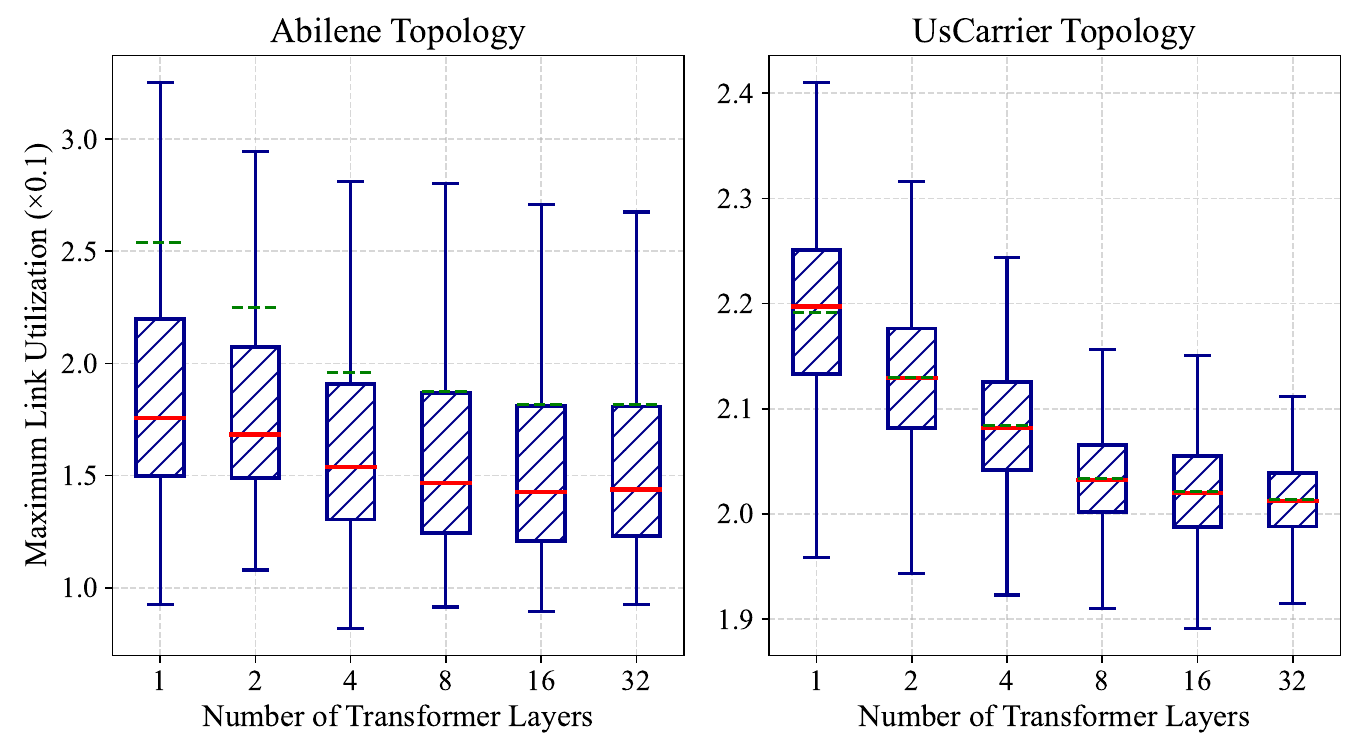}
    }
    \caption{Exploring the importance of backbone and layers of pre-trained LMs.}
\end{figure}

\paragraph{Impact of LMs' Layers}
To measure the impacts of LMs' Transformer layers on the WAN TE performance, we conduct an experiment across varying layer configurations. In Figure~\ref{fig:layers}, we note the findings: (i) more layers generally lead to better performance, but the gains plateau beyond a certain depth, and (ii) as the topology size increases, the optimal number of layers tends to be larger, though typically no more than ten layers are required. Overall, this is consistent with our theoretical findings (see Theorem~\ref{theorem:3}). 

\begin{wrapfigure}[14]{r}{0.72\textwidth}
  \vspace{-0.5cm}
  \begin{center}
  \includegraphics[width=1.\linewidth]{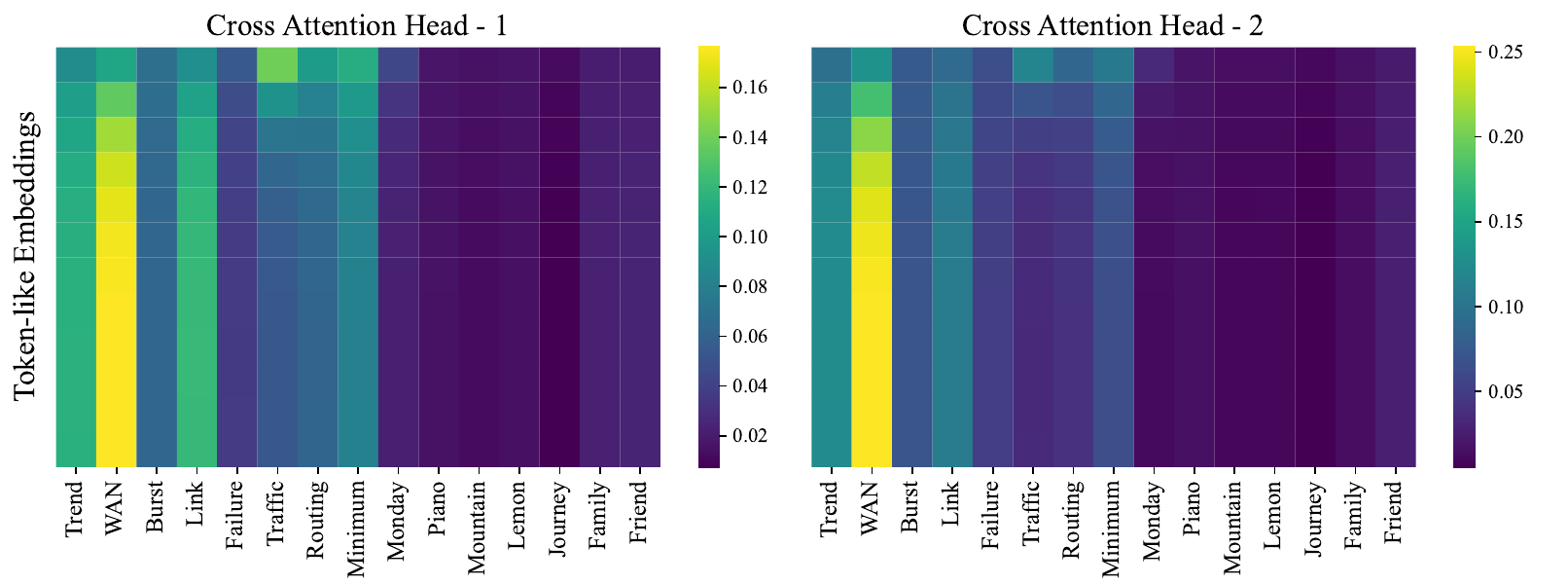}
\caption{Learned cross-attention maps from the input alignment module.}
\label{fig:representation}
\end{center}
\end{wrapfigure}

\paragraph{Representation Interpretation}\label{subsec:visulization}
We next visualize the cross-attention maps from two randomly selected attention heads in the embedding-to-language alignment module of \textsc{Lmte} in Figure~\ref{fig:representation}, based on the G\'{E}ANT dataset. Each row corresponds to a random spatio-temporal input instance, and each column denotes a chosen token embedding, covering both TE-relevant terms (e.g., ``$\mathsf{WAN}$'', ``$\mathsf{Link}$'') and TE-irrelevant ones (e.g., ``$\mathsf{Piano}$'', ``$\mathsf{Friend}$''). The heatmaps reveal strong alignment between input features and semantically relevant tokens, suggesting that the module bridges the gap between heterogeneous modalities. Our proposal is thus interpretable.

\begin{figure}
\centering
    \subfigure[Tokenization (input) module\label{fig:input}]{
        \centering
        \includegraphics[width=0.45\linewidth]{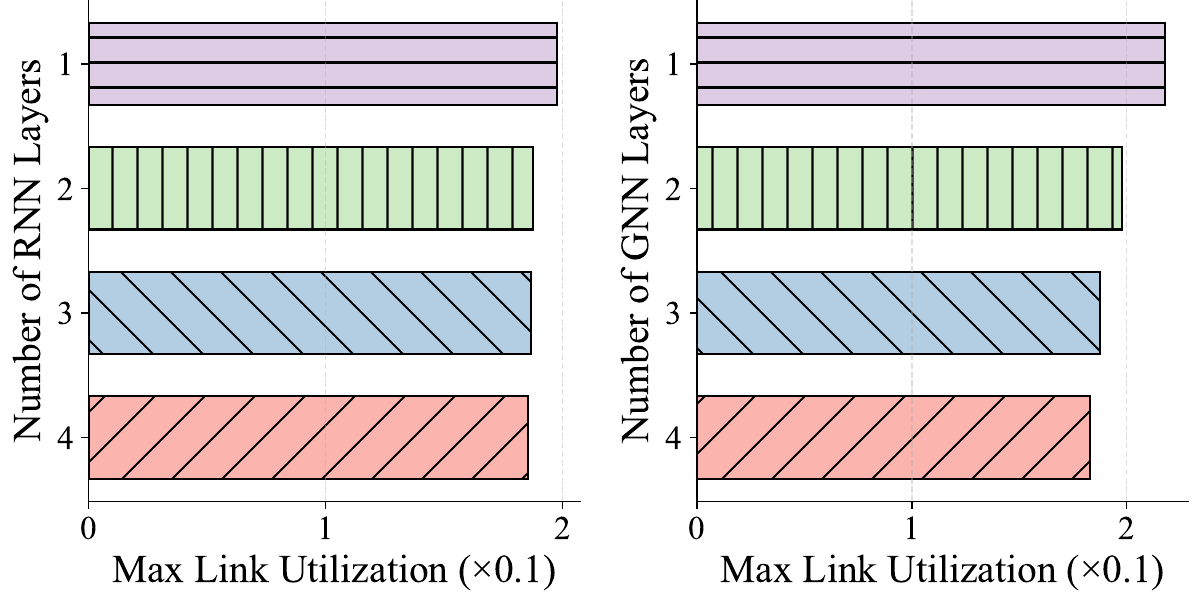}
    }
    \subfigure[Detokenization (output) module\label{fig:output}]{
        \centering
        \includegraphics[width=0.45\linewidth]{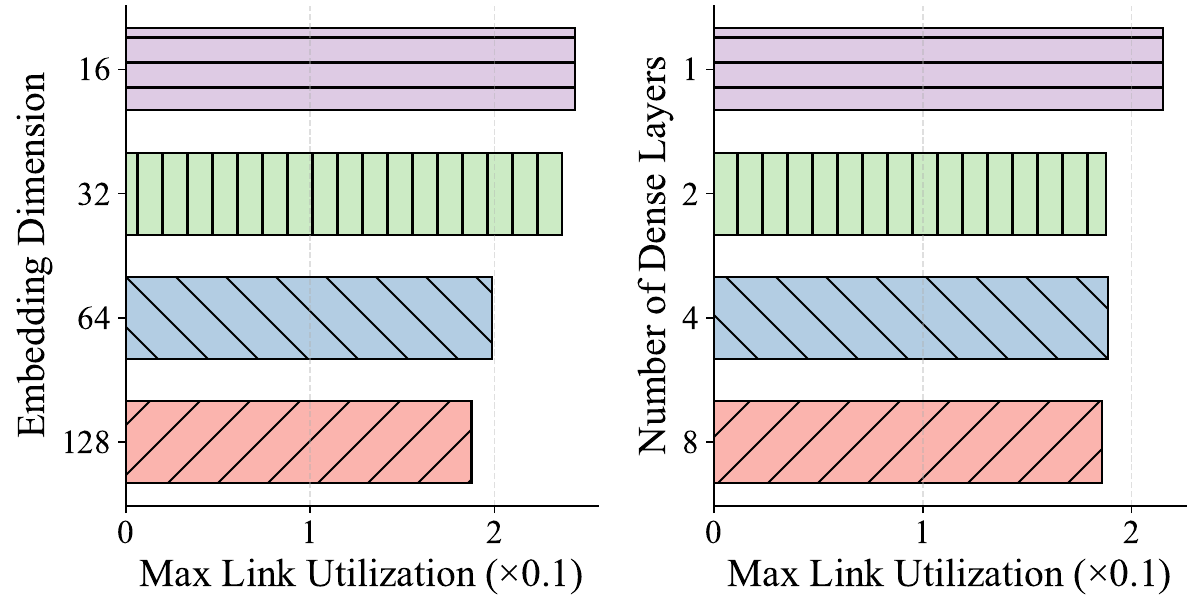}
    }
    \caption{Sensitivity analysis of \textsc{Lmte}’s hyperparameters on Abilene.}
\end{figure}

\paragraph{Sensitivity Analysis}
We further analyze the sensitivity of \textsc{Lmte}’s performance to key hyperparameters. Specifically, Figure~\ref{fig:input} examines the impact of varying the number of RNN and GNN layers, while Figure~\ref{fig:output} explores the effects of changing the number and dimension of dense layers. We can observe that: \circled{1} With respect to input alignment, RNN complexity plays a less significant role than GNN complexity, as large-scale topologies generally necessitate deeper GNN layers to capture their structural information. \circled{2} Regarding output decision-making, the width of a neural network has a greater impact than its depth. For instance, even a two-layer MLP can approximate the optimal solution, while increasing the hidden size continues to improve performance.

\begin{figure}
    \centerline{\includegraphics[width=0.81\linewidth]{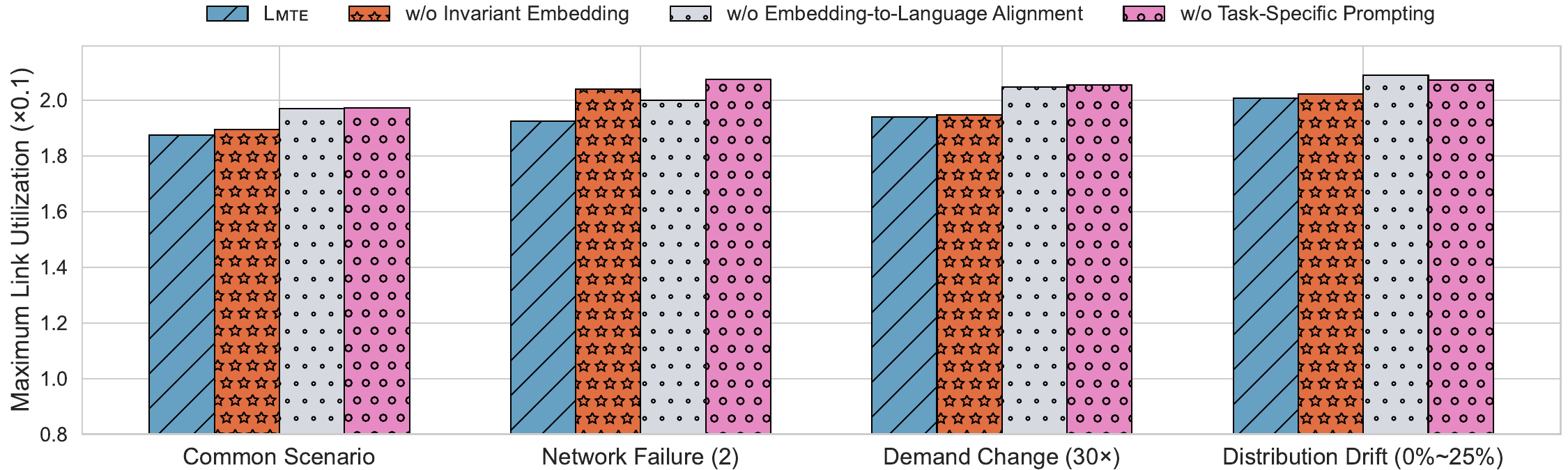}}
    \caption{Ablation study of \textsc{Lmte}’s three key designs on Abilene.}
    \label{fig:ablation_study}
\end{figure}

\paragraph{Ablational Study}
Figure~\ref{fig:ablation_study} presents an ablation study evaluating the impact of \textsc{Lmte}’s key components (i.e., embedding in \S\ref{comp1}, alignment in \S\ref{comp2} and prompting in \S\ref{comp3}) on its overall performance. We find that each component of the proposed framework contributes meaningfully to enhancing its TE quality, either in one scenario or across multiple scenarios. First, \textsc{Lmte} w/o Invariant Embedding exhibited a 10\% increase in MLU upon 2 link failures, demonstrating that our GNN-based embedding approach can effectively handle topological changes. Second, \textsc{Lmte} w/o Embedding-to-Language Alignment resulted in an average performance degradation of 11\%, which underscores the importance of modality alignment. Third, \textsc{Lmte} w/o Task-Specific Prompting led to an 8\%–15\% increase in MLU, because the language model primarily understands the task and scenario through the input natural language (i.e., the prompt). This demonstrates that our design is interpretable.

\section{Conclusion}
\label{final_sec}

This work provided an initial yet comprehensive exploration of pre-trained language models for WAN TE, aiming to bring advanced reasoning capabilities to the next choice of high-quality configurations both in theory and in practice. We developed \textsc{Lmte}, the first LM-driven TE framework that combines carefully designed, highly interpretable components. 
\textsc{Lmte} produces first-rate routing decisions while offering superior generalization compared to SOTA TE approaches. Though not the final answer, \textsc{Lmte} marks a meaningful step toward sustainable networking solutions powered by LMs. 
We hope this new paradigm can promote future network optimization research. The authors have provided public access to their code and data at \href{https://github.com/Y-debug-sys/LMTE}{\texttt{https://github.com/Y-debug-sys/LMTE}}.

\section{Acknowledgments}

The work was supported by the Natural Science Foundation of China (No. 62572168, 62372149, U23A20303, 92373205), Key Laboratory of Knowledge Engineering with Big Data (the Ministry of Education of China) (BigKEOpen2025-04), Anhui Provincial Natural Science Fund (2508085MF151), Key Research and Development Program of Zhejiang Province under Grant (No. 2025C02103), and National Key Research and Development Program of China (No. 2023YFB4404401).


\bibliographystyle{unsrt}  
\bibliography{main}  






\newpage

\appendix

\part*{Appendix} 

\section{Notations \& Useful Lemmas}

\begin{definition}
    \textbf{(Two-Commodity Integral Flow, D2CIF)} Given a directed graph $\mathcal{G} = (\mathcal{V}, \mathcal{E})$ with edge capacities $c_e \in \mathbb{N}$ for all $e \in \mathcal{E}$, two source-terminal pairs $(s_1, t_1)$ and $(s_2, t_2)$, and a positive integer requirement $R$, \textit{determine whether there exist integral flow functions} $f_1, f_2 \colon \mathcal{E} \to \mathbb{N}$ such that the following conditions hold: (i) The capacity constraint $f_1(e) + f_2(e) \leq c_e$ is satisfied for all $e \in \mathcal{E}$, and (ii) for each commodity $i \in \{1, 2\}$, the net flow $F_i$ out of source $s_i$, defined as $F_i = \sum_{\substack{e \in \mathcal{E}: \mathsf{head}(e) = s_i}} f_i(e) - \sum_{\substack{e \in \mathcal{E}: \mathsf{tail}(e) = s_i}} f_i(e)$, satisfies the total flow requirement $F_1 + F_2 = R$.
    \label{define:1}
\end{definition}

\begin{definition}
    \textbf{(Useful Notations for Analytical Optimality)} Let $\mathcal{H}$ denote the finite set of all possible history realizations. Let state $\pi : \mathcal{H} \to \mathcal{R}$ be a mapping from the finite set of all possible history realizations to a TE configuration. We define the expected loss (e.g. expected MLU) of a TE state $\pi$ as $\tilde{\mathcal{L}}(\pi) = \mathbb{E} \left[ \mathcal{L}(\pi, \mathcal{D}_t) \right]$, and let the optimal TE state be $\pi^* \in \arg\min_\pi \tilde{\mathcal{L}}(\pi)$.
    \label{define:2}
\end{definition}

\begin{lemma}
    WAN traffic engineering for MLU is at least as hard as the Two-Commodity Integral Flow problem.
    \label{lemma:2}
\end{lemma}
\begin{proof}    
    We reduce the D2CIF decision problem to a special case of the MLU optimization problem. Set up the TE MLU problem with fixed total demand $R$ and minimize the utilization $\alpha = \max_{e \in \mathcal{E}} \frac{f_e}{c_e}$. Then the key observation is: (i) If \( \alpha \leq 1 \), then no edge exceeds its capacity, meaning the routing is a feasible solution to D2CIF that sends full demand \( R \). (ii) If \( \alpha > 1 \), then at least one edge is overloaded — meaning it's impossible to satisfy the D2CIF requirement without violating capacity. Therefore, the decision version of the MLU problem: "Is the minimum possible MLU \( \leq 1 \)?" is exactly equivalent to solving D2CIF. But the full MLU problem goes further — it not only decides feasibility, but also optimizes to find the least congested way of routing total flow \( R \). Hence, the MLU optimization problem is at least as hard as D2CIF.
\end{proof}

\begin{lemma}
    The Two-Commodity Integral Flow problem is NP-complete, even when all edges simply have capacity 1.
    \label{lemma:3}
\end{lemma}
\begin{proof}
    The proof of Lemma~\ref{lemma:3} can be concluded from Theorem 3 of~\cite{even1975complexity}, where the NP-completeness of Simple D2CIF is shown via reduction from 3SAT (Boolean 3-Satisfiability Problem~\cite{cook1971complexity}). The construction maps variables to lobe structures (with upper/lower paths as assignments) and clauses to sinks, where achieving \(1+k\) flow (for \(k\) clauses) corresponds to 3SAT satisfiability: the first commodity's path choice encodes assignments while the second's \(k\) units validate clauses. This solution-preserving reduction establishes the result.
\end{proof}

\begin{lemma}
    There exists a finite number of transitions $T$ from any arbitrary initial state, such that the expected suboptimality of the average state $\bar{\pi} = \frac{1}{T} \sum_{t=1}^{T} \pi_t$ satisfies:
    \begin{equation}
        \left| \tilde{\mathcal{L}}\left( \bar{\pi} \right) - \tilde{\mathcal{L}}(\pi^*) \right| \le \varepsilon, \quad \forall \varepsilon > 0.
    \end{equation}
    \label{lemma:4}
\end{lemma}
\begin{proof} 
    By leveraging the convexity of the TE objectives, the lemma’s proof implies a stochastic gradient descent (SGD~\cite{shapiro2021lectures})-style transition rule for convergence to the optimal solution: starting from arbitrary initial states, the states are iteratively adjusted in the negative gradient direction until reaching the global minimum. Formally, the state transition is defined as:  
    \[
        \pi_t = \delta(\pi_{t-1}, \sigma_t) := \text{Proj}\left\{ \pi_{t-1} - \eta v_t \right\},  
    \] 
    where \( v_t \in \partial \mathcal{L}(\pi_t, \mathcal{D}_t) \) is a subgradient of the TE objective, \(\text{Proj}\{\cdot\}\) denotes the projection onto the simplex for each OD pair, and \(\eta\) is the step size. Then according to Theorem 1 of~\cite{perry2023dote}, we have: For any target accuracy \(\varepsilon > 0\), there exist a learning rate \(\eta > 0\) and a finite iteration horizon \(K\) such that $\left| \tilde{\mathcal{L}} \left[ \mathbb{E} \left( \frac{1}{K} \sum_{k=1}^{K} \pi_k \right) \right] - \tilde{\mathcal{L}}(\pi^*) \right| \le \varepsilon$, where the expectation is taken with respect to the stochastic sampling process of the algorithm. Since the batch size equals 1 in our scenario, we have $\mathbb{E} \left( \frac{1}{K} \sum_{k=1}^{K} \pi_k \right) = \frac{1}{K} \sum_{k=1}^{K} \pi_k$. Therefore, there exists a finite number of gradient-based transitions $T$ such that $\left| \tilde{\mathcal{L}}\left( \bar{\pi} \right) - \tilde{\mathcal{L}}(\pi^*) \right| \le \varepsilon$, which concludes the proof.
\end{proof}

\begin{lemma}
    (Universal Approximation Theorem~\cite{cybenko1989approximation,pinkus1999approximation}) Let \( K \subset \mathbb{R}^d \) be a compact set, and let
    \[
        \mathcal{H}=
            \left\{ x \mapsto \sum_{i=1}^{m} a_i , \sigma(w_i^\top x + b_i);\middle|; m \in \mathbb{N},\ a_i \in \mathbb{R},\ w_i \in \mathbb{R}^d,\ b_i \in \mathbb{R} \right\}
    \]
    denote the class of single-hidden-layer neural networks with activation \( \sigma : \mathbb{R} \to \mathbb{R} \). If \( \sigma \) is continuous and non-polynomial, then \( \mathcal{H} \) is dense in \( C(K) \) under the uniform norm. That is, for any \( f \in C(K) \) and any \( \varepsilon > 0 \), there exists \( h \in \mathcal{H} \) such that
    \(
       \sup_{x \in K} |f(x) - h(x)| < \varepsilon.
    \)
    \label{lemma:5}
\end{lemma}
\begin{proof}
    We prove the result by contradiction. Assume that \( \mathcal{H} \) is not dense in \( C(K) \). Then there exists a nonzero continuous linear functional \( L : C(K) \to \mathbb{R} \) such that
    \begin{equation}
        L(h) = 0, \quad \forall h \in \mathcal{H}.
        \label{tag1}
    \end{equation}
    By the Riesz representation theorem, there exists a finite signed Borel measure \( \mu \) on \( K \) satisfying
    \begin{equation}
        L(g) = \int_K g(x), d\mu(x), \quad \forall g \in C(K).
        \label{tag2}
    \end{equation}
    Combining (\ref{tag1}) and (\ref{tag2}), we obtain
    \begin{equation}
        \int_K \sigma(w^\top x + b), d\mu(x) = 0, \quad \forall w \in \mathbb{R}^d,\ b \in \mathbb{R}.
        \label{tag3}
    \end{equation}
    Since \( \sigma \) is continuous and non-polynomial, the family \( {\sigma(w^\top x + b)}_{w,b} \) is discriminatory in the sense of \cite{cybenko1989approximation}, implying that the only finite signed measure satisfying (\ref{tag3}) is the zero measure \( \mu \equiv 0 \). This contradicts the assumption that \( L \) is nonzero. Therefore, \( \mathcal{H} \) is dense in \( C(K) \), completing the proof.
\end{proof}

\section{Proof of Lemma~\ref{lemma:1}}
\label{app:lem1}

\begin{proof}
    At a high level, the proof of Theorem~\ref{lemma:1} consists of two key components: (i) establishing the NP-hardness of the original TE formulation, and (ii) constructing a finite-step automata that guarantees convergence to the optimal configuration from arbitrary initial states. We start by making the observation told by Lemma~\ref{lemma:2} about connection between traffic engineering and Two-Commodity Integral Flow (D2CIF). Next, we can use Lemma~\ref{lemma:3} to show the NP-Completeness of simple D2CIF. Therefore, it follows that the TE problem under investigation is NP-hard. However, our second key concern is whether there exists a finite-length automaton state transition path that yields an optimal solution. The automata's operation circumvents NP-hardness, as presented in Lemma \ref{lemma:4}, by leveraging local gradient information \(v_t \in \partial \tilde{\mathcal{L}}(\pi_t, \mathcal{D}_t)\) and simplex projections to maintain feasibility, avoiding direct solution of the combinatorial problem. The finite-time \(\varepsilon\)-guarantee holds uniformly across initial states due to the TE objectives' convexity, which enables stochastic subgradient methods to bypass discrete flow allocation's inherent hardness. Interpreted as a Markov process, the automata's states \(\pi_t\) converge ergodically around \(\pi^*\). Thus, while exact solutions remain NP-hard, any TE instance can be efficiently approximated via finite automata simulations (Lemma \ref{lemma:4}), concluding the proof.
\end{proof}

\section{Proof of Theorem~\ref{theorem:2}}
\label{app:the2}

\begin{proof}[Proof Sketch]
    Adapted from Yun et al.~\cite{yun2020transformers}, the proof of Theorem~\ref{theorem:2} goes in three steps: (i) Approximate state-to-state mappings with step functions on compact domain. A step function is piecewise constant and can approximate continuous functions by leveraging Lemma~\ref{lemma:5}. Particularly, for any \( f \), there exists piece-wise constant function \( \tilde{f} \) such that \( d(f, \tilde{f}) \leq \epsilon/3 \). (ii) Approximate piece-wise constant functions with hardmax-based Transformers. Using hardmax operator simplifies the construction of piece-wise constant functions. According to Proposition~4 of~\cite{yun2020transformers}, for all \( \tilde{f} \), there exists the modified Transformer $g$ such that \( d(\tilde{f}, g) \le C_{\tilde{f}} \), where \( C_{\tilde{f}} \) is a constant associated with \( \tilde{f} \). (iii) Approximate hardmax-based Transformers with original Transformers. Due to softmax approximating hardmax (see Lemma~4 of~\cite{liu2023transformers}), we also have that the modified Transformer can be approximated with error \( \leq \epsilon/3 \). Finally, by choosing \( C_{\tilde{f}} \leq \epsilon/3 \) and applying the triangle inequality, we obtain  $d(f, h) \le {d(f, \tilde{f})} + {d(\tilde{f}, g)} + {d(g, h)} \le \epsilon$.
\end{proof}

\section{Proof of Theorem~\ref{theorem:3}}
\label{app:the3}

\begin{proof}
    At a high level, the proof of Theorem~\ref{theorem:3} consists of two parts: the composed transition chain \(\delta(\cdot, \sigma_T) \circ \cdots \circ \delta(\cdot, \sigma_1)\) can be evaluated in $\mathcal{O}\left(\log T\right)$-depth via a binary tree, and the shortcuts can be constructed with a high probability. The first part is a direct application of Theorem~1 of~\cite{liu2023transformers}. We here focus on proving the second part. Without loss of generality, we assume the TE state is a 1-dimensional vector for clarity.
    
    Define $S=\mathcal{O}\left( \log T \right) \ge 1$, sequential state vector $u \in \mathbb{R}^T$, random matrix $R \in \mathbb{R}^{S \times T}$ in which $R_{i,j} \sim \mathcal{N}\left( 0, 1 \right)$, and shortcut state vector $v \in \mathbb{R}^S$ in which $v_i = \frac{1}{\sqrt{S}} \sum_j R_{i,j}u_j$. Then let $x = \frac{\sqrt{S}}{\|u\|_2} v$ i.e. $x_i = \frac{R_i^T u}{\|u\|_2}$, we have $\|x\|_2^2 = \sum_{i=1}^S x_i^2 = \frac{S\|v\|_2^2}{\|u\|_2^2}$ and $x \sim \mathcal{N} \left(0, I_S\right)$. By Markov’s inequality i.e. $\Pr \left[ x \ge a \right] \le \frac{\mathbb{E}\left( x \right)}{a}$, for any $0 < \epsilon < 1$ and $0 < \lambda < 0.5$, we can write
    \[
        \begin{aligned}
           \Pr\left[\|v\|_2^2 \ge \left( 1 + \epsilon \right) \|u\|_2^2\right] = \Pr\left[ \frac{\|u\|_2^2}{S} \cdot \|x\|_2^2 \ge \left( 1 + \epsilon \right) \|u\|_2^2\right] &= \Pr\left[ e^{\lambda \|x\|_2^2}  \ge e^{\left( 1 + \epsilon \right) \lambda S} \right] \\ 
           &\le \frac{\mathbb{E}\left[ e^{\lambda \|x\|_2^2} \right]}{e^{\left( 1 + \epsilon \right) \lambda S}} = \left[ \frac{\mathbb{E}\left( e^{\lambda x_i^2} \right)}{e^{\left( 1 + \epsilon \right) \lambda S}} \right]^S
        \end{aligned}
    \]
    Using MGF of $\chi^2$ distribution i.e. $\mathbb{E}\left( e^{\lambda \chi_1^2} \right) = \frac{1}{\sqrt{1-2\lambda}}$, and setting $\lambda=\frac{\epsilon}{2\left( 1+\epsilon \right)}$, we can write
    \[
        \left[ \frac{\mathbb{E}\left( e^{\lambda x_i^2} \right)}{e^{\left( 1 + \epsilon \right) \lambda S}} \right]^S \le \left[ \frac{1}{\sqrt{1-2\lambda} \cdot e^{\left( 1 + \epsilon \right) \lambda}} \right]^S = e^{{S \left( \log \left( 1 + \epsilon \right) - \epsilon \right)}/{2}}
    \]
    Now using the inequality $\log \left( 1 + a \right) < a - a^2/2 + a^3/3$, and selecting an appropriate $S=\mathcal{O}\left( \log T \right)$,
    \[
       \Pr\left[\|v\|_2^2 \ge \left( 1 + \epsilon \right) \|u\|_2^2\right] \le e^{-2\log T} = \frac{1}{T^2}
    \]

    Similarly, by setting $\lambda = \frac{\epsilon}{2\left( 1 - \epsilon \right)}$, we get
    \[
        \begin{aligned}
           &\Pr\left[\|v\|_2^2 \le \left( 1 - \epsilon \right) \|u\|_2^2\right] = \Pr \left[ e^{-\lambda \| x \|_2^2} \ge e^{-\left( 1 - \epsilon \right) \lambda S} \right] \le \frac{\mathbb{E}\left( e^{-\lambda \| x \|_2^2} \right)}{e^{-\left( 1 - \epsilon \right) \lambda S}} \le \left[ \frac{1}{\sqrt{1+2\lambda} \cdot e^{\left( 1 - \epsilon \right) \lambda}} \right]^S = e^{S \left( \log \left( 1-\epsilon \right) - \epsilon \right) / 2}
        \end{aligned}
    \]
    Using the inequality $\log \left( 1-a \right) < -a - a^2/2$ and selecting an appropriate $S=\mathcal{O}\left( \log T \right)$,
    \[
        \Pr\left[\|v\|_2^2 \le \left( 1 - \epsilon \right) \|u\|_2^2\right] \le e^{-2\log T} = \frac{1}{T^2}
    \]
    Combining the above, we can obtain 
    \begin{equation}
        \Pr \left[ \left( 1-\epsilon \right) \|u\|_2^2 \le \|v\|_2^2 \le \left( 1+\epsilon \right) \|u\|_2^2 \right] > 1-\frac{2}{T^2}
        \label{JL_Lemma}
    \end{equation}
    Finally, by setting \( u = X_i - X_j \) and \( v = f(X_i) - f(X_j) \), where \( X_i, X_j \in \mathbb{R}^T \) and \( f: \mathbb{R}^T \rightarrow \mathbb{R}^S \) is the shortcut mapping, Eqn.~\ref{JL_Lemma} implies that pairwise distances are preserved under the mapping with high probability (close to 1 even for modest $T$). The model can thus simulate state transitions in a compressed space of dimension \( \mathcal{O}(\log T) \), completing our proof.
\end{proof}

\end{document}